\documentclass[journal]{IEEEtran}

\usepackage{moreverb}
\usepackage{epsfig}
\usepackage{amsmath,amssymb,amsthm,mathrsfs,amsfonts,dsfont}
\usepackage{adjustbox,lipsum}
\usepackage[linesnumbered,ruled,vlined]{algorithm2e}
\usepackage{amsfonts}
\usepackage{epsfig}
\usepackage{amssymb}
\usepackage{amsmath}
\usepackage{amsthm}
\usepackage{multirow}
\usepackage{setspace}
\usepackage{rotating}
\usepackage{graphicx}
\usepackage{tabularx}
\usepackage{array}
\usepackage{anyfontsize}
\usepackage{color,soul}
\usepackage{graphicx,dblfloatfix}
\usepackage{epstopdf}
\usepackage{blindtext}
\usepackage{amsmath}
\usepackage{amsthm,amssymb,amsmath,bm}
\usepackage{subfigure}
\usepackage{amsfonts}
\usepackage{epsfig}
\usepackage{amssymb}
\usepackage{amsmath}
\usepackage{cite}
\usepackage{graphicx}
\usepackage{fancyhdr}
\usepackage{caption}
\usepackage{tabularx}
\usepackage{tcolorbox}
\usepackage{cite}
\usepackage{setspace}

\usepackage[inline]{enumitem}

\DeclareMathOperator*{\argmin}{arg\,min}
\usepackage{texdef2020}
\newcommand{\pb}{P_B}

\newtheorem{theorem}{Theorem}
\newtheorem{lemma}{Lemma}
\newtheorem{corollary}{Corollary}

\newtheorem{definition}{Definition}
\newtheorem{remark}{Remark}

\def\blue{\textcolor{blue}}

\allowdisplaybreaks

\newcommand{\ZW}[1]{\text{Zero-Wait-{#1}}}
\newcommand{\ZWone}{\ZW{1}}
\newcommand{\ZWtwo}{\ZW{2}}
\newcommand{\Wone}{\text{Wait-1}}
\newcommand{\Agelimit}[1]{[#1]^{\bar\Delta}}

\title{Status Update Control and Analysis under Two-Way Delay}

\author{
\IEEEauthorblockN{Mohammad~Moltafet, Markus~Leinonen, Marian~Codreanu, and Roy~D.~Yates 
}

\thanks{Mohammad Moltafet is with the Department of Electrical and Computer Engineering, University of California Santa Cruz (UCSC), Santa Cruz, CA 95064, USA (e-mail: mmoltafe@ucsc.edu), Markus Leinonen, deceased, was with Centre for Wireless Communications--Radio Technologies, University of Oulu,
 Oulu, Finland,
Marian Codreanu is with the Department of Science and Technology, Link\"{o}ping University, Link\"{o}ping 58183, Sweden (e-mail: marian.codreanu@liu.se), and Roy D. Yates is with the Wireless Information Network Laboratory, ECE
Department, Rutgers University, North Brunswick, New Jersey 08902, USA (e-mail:
ryates@winlab.rutgers.edu).
}
}  

\begin{document}
\maketitle
\sloppy

\begin{abstract}
We study status updating under two-way delay in a system consisting of a sampler, a sink, and a controller residing at the sink. The controller drives the sampling process by sending request packets to the sampler. Upon receiving a request, the sampler generates a sample and transmits the status update packet to the sink. Transmissions of both request and status update packets encounter random delays. We develop optimal control policies to minimize the average age of information (AoI) using the tools of Markov decision processes in two scenarios. We begin with the system having at most one active request, i.e., a generated request for which the sink has not yet received a status update packet. Then, as the main distinctive feature of this paper, we initiate pipelined-type status updating by studying a system having at most two active requests. Furthermore, we conduct AoI analysis by deriving the average AoI expressions for the $\ZW{1}$, $\ZW{2}$, and $\Wone$ policies. According to the $\ZW{1}$ policy, whenever a status update packet is delivered to the sink, a new request packet is inserted into the system. The $\ZW{2}$ policy operates similarly, except that the system can hold two active requests. According to the $\Wone$ policy, whenever a status update packet is delivered to the sink, a new request is sent after a waiting time which is a function of the current AoI. Numerical results illustrate the performance of each status updating policy under varying system parameter values.

\end{abstract}

\section{Introduction}
A key enabler for emerging time-critical applications (e.g., industrial automation or drone control) in wireless sensor networks (WSNs) and Internet of things (IoT) is the timely delivery of status updates of real-world physical processes. The age of information (AoI) is a relatively new destination-centric metric 
for the information freshness in status update systems \cite{5984917,Sunbook2019,9380899}. A status update packet contains a measured value of a monitored process and a time stamp
for when the sample was generated. 
For each source node, the AoI at a destination measures the difference between the current time and the time stamp of the most recent
received status update \cite{5984917,6195689}.

In this paper, we study status updating under two-way delay in the discrete-time system depicted in Figure~\ref{Model}. 
A sampler monitors a random process and sends updates through a ``forward link'' service facility to a sink/monitor. 
At the sink, a controller drives the sampling process by sending request packets to the sampler
on a ``reverse link'' control channel.

This system is subject to two-way delays in that
packets on both forward and reverse links experience random delays. In particular, each request packet on the reverse link  requires a random time to be delivered to the sampler. Upon receiving a request, the sampler immediately generates a sample and initiates the transmission of the corresponding status update packet which requires a random time to arrive at the sink. 
This model arises in many applications, where the monitor (sink) is responsible for controlling the sampling process of a remote source (e.g., a wireless sensor); accordingly, the monitor sends sampling control messages (i.e., packets) over the wireless links, which arrive at the source after a random delay. Such delayed control messaging is not present in the majority of the existing literature on AoI analysis and control; the few exceptions 
are elaborated in Section \ref{sec_related}.

\subsection{Contributions}

The delayed control channel opens up a framework to study status updating under \textit{pipelining-type} packet flow, where multiple requests are concurrently being processed in the system. In this setting, we define an {\em active request} as a request that the controller has created, but for which the sink has not yet received the corresponding status update packet.
As the main contribution, we take a step towards improving the timeliness of status updates in pipelining-based status update systems as follows.
We initially consider a system with at most one active request, {thus operating in a stop-and-wait fashion.}
Then, initiating pipelining-type status updating, we consider a system that supports two active requests. As this system leads to queued packets when a server is busy, waiting buffers 
are introduced at both forward and reverse links. 

The first part of this paper deals with AoI-optimal status updating \textit{control} in the two scenarios regarding the number of active requests. Assuming that the controller has global knowledge of the system state (i.e., the existence of packets at the waiting buffers and servers and their age values), our objective is to find the optimal times to send request messages, i.e., the optimal control policy, that minimize the average AoI. For each scenario, we formulate a Markov decision process (MDP) problem that we subsequently solve via the relative value iteration (RVI) algorithm. 

The second part of this paper focuses on \textit{AoI analysis}. Since, in practice, the global system state is oftentimes not available at the controller, we analyze the performance of two control policies that operate solely based on the controller's local information: 1) the zero-wait policy with one active request ($\ZWone$ policy), and 2) the zero-wait policy with two active requests ($\ZWtwo$ policy). According to the $\ZWone$ policy, whenever a status update packet is delivered to the sink, a new request packet is immediately inserted into the system. The $\ZWtwo$ policy operates similarly; except the system always maintains two active requests.  We derive the closed-form expression of the average AoI under the $\ZWone$ and $\ZWtwo$ policies. 

Furthermore, driven by the fact that the zero-wait discipline is not necessarily age-optimal \cite{7283009,8000687}, we modify the $\ZWone$ policy by introducing a waiting time in the controller's update request process. Under this $\Wone$ policy, whenever a status update packet is delivered to the sink, a new request is sent after a waiting time which is a function of the current AoI. We derive the closed-form expression of the average AoI under the $\Wone$ policy. 

Finally, numerical results are presented to show the performance of each status updating policy under different system parameter values. Moreover, the threshold-type structure of an optimal status updating policy is illustrated.

\subsection{Organization}
The paper is organized as follows. Related work is discussed in Section~\ref{sec_related}. The system model and AoI performance metric are presented in Section~\ref{System Model}. Optimal control policies to minimize the average AoI are presented in Section~\ref{Optimal Control for AoI Minimization}. Average AoI analysis for 
$\ZW{1}$, $\ZWtwo$, and
$\Wone$ policies is  provided in Section~\ref{Average AoI Analysis}. 
Numerical results are presented in Section~\ref{Numerical Results}. Finally, concluding remarks are made in Section~\ref{Conclusions}.

\section{Related Work}\label{sec_related}
Analysis of AoI in various status update systems has witnessed great interest. 
The first analytical results were derived in \cite{6195689}, where the authors characterized the average AoI for M/M/1, D/M/1, and M/D/1 first-come first-served queueing models. In \cite{7415972}, the authors proposed the peak AoI as an alternative metric for information freshness. Since then, average AoI and peak AoI have extensively been studied in various queueing models. 
The average (peak) AoI for various queueing models with exponentially distributed service time and Poisson arrivals has been studied in, e.g.,
\cite{6310931,7541765,7415972,9048933,7364263,6284003}.
Besides exponentially distributed service time and Poisson arrivals, AoI has also been studied under various arrival processes and service time distributions in, e.g.,
\cite{9119460,8406909,8006504,8820073,soysal2019age,9048933,9048909,7282742,8886357}.
Using the framework of stochastic hybrid systems \cite{Hespanha-shs-06}, the average AoI and moment generating functions of the AoI have been characterized under various queueing models and packet management policies in, e.g.,   \cite{8469047,9103131,8437591,8406966,8437907,9013935,9048914,9252168,9162681,Moltafet2020mgf,9174099,9562231,MajedMGF,9611498}. 
AoI has also been studied in various discrete-time queueing systems \cite{7249268,8764468,9120608,9148775,9611393}.

Besides the age analysis, numerous works have focused on devising AoI-optimal status updating control procedures. Indeed, the optimization of, e.g., the sampling times of the sensors, scheduling, and resource management plays a critical role in the performance of status update systems. The works \cite{8772205,bhat2019throughput,8943134,8723545,9181539,gu2020optimizing,9155420,8445873,8764465,bedewy2021optimal,8000687,chen2021optimal,8648525,9598864,zakeri22,9540757,arxiv02929,9834697,9868923} investigated the optimal sampling problem in status update systems, whereas the works \cite{7283009,8006703,8406974,8437573,8123937,8822722,8606155,9546792} studied the optimal sampling problem with energy harvesting sources. A comprehensive literature review of recent works in AoI can be found in \cite{9380899}.

An efficient status update system optimally balances between too infrequent updates and excessively frequent updates that induce queueing delays. To this end, an intuitive approach is to send new updates as the system becomes empty. However, the work \cite{7283009} established a new notion for optimal status update control by showing that a lower average AoI value may be achieved by purposely asking the sensor to wait before taking and transmitting a fresh sample. The work \cite{8000687} generalized the idea in \cite{7283009} and studied optimal sampling with respect to different age penalty functions, where a source communicates over a Markov channel (imposing correlated packet transmission times) subject to a sampling rate constraint. The work \cite{8445873} proposed a new mutual-information-based metric to quantify information freshness and subsequently studied a sampling problem, where a Markov source sends updates through a queue to the sink. 
The authors proved that the optimal sampling policy is a threshold policy and found the optimal threshold. The work \cite{8764465} further extended the study of the sampling problem in \cite{8000687} by considering more general non-linear functions of AoI and investigating both continuous-time and discrete-time systems. 

In \cite{9868923}, the authors considered a status update system consisting of an information source, a server, and multiple users. The server
communicates with the information source to update its data, whereas users submit requests to the server to retrieve the updated data. The main goal of the server is to provide users with fresh information about the information source. However, maintaining fresh data at the server requires the server to update its database frequently, which incurs an update cost. They studied the tradeoff between the AoI and the update cost at the server.
The work \cite{8006703} considered an energy harvesting sensor and derived the optimal threshold, in terms of the current energy state and estimated age, to trigger a new sample to minimize the average AoI. Therein too, the threshold policy conforms to the idea of intentionally waiting prior to sampling, as in \cite{7283009}.

Extending beyond mere sampling, status update control has also been studied in more general network optimization frameworks. The authors of \cite{bedewy2019ageoptimal-conf,bedewy2021optimal} considered a multi-source status update system in which the different sources send their updates through a shared channel causing random delays. They studied joint sampling and transmission scheduling for optimizing the information freshness. The authors of \cite{chen2021optimal} considered two source nodes generating heterogeneous traffic with different power supplies and studied the peak-age-optimal status update scheduling. In \cite{8406974}, the authors considered an energy harvesting source sending status updates over a noisy channel and derived the optimal status updating policy to minimize the average AoI. In \cite{8648525}, the authors studied transmission scheduling of status updates over an error-prone communication channel via automatic repeat request (ARQ) protocols to minimize the average AoI subject to the constraint on the average number of transmissions. The work \cite{9598864} studied the problem of optimal scheduling and radio resource allocation in a WSN to minimize the average total transmit power subject to an average AoI constraint. The work in \cite{zakeri22} studied the AoI-optimal sampling and scheduling in a relay-based multi-source status update system under an average number of transmissions constraint. None of the aforementioned works, however, considers a delayed control channel.

The most closely related works to this paper are \cite{9540757,arxiv02929,9834697}; these contributions all incorporate a delayed feedback/ACK channel in a status update system. The work in \cite{9540757} considered a system consisting of a sensor, a controller, a forward sensor-to-controller channel, and a backward controller-to-sensor channel. A transmission in both the forward and backward channels experiences a random delay. A sensor can generate a new packet only after receiving a confirmation message of the previous packet from the controller, thus conforming to a \textit{stop-and-wait} ARQ mechanism. The authors introduced waiting times (in line with \cite{7283009,8000687}) at both the sensor and controller ends and derived optimal status updating policies to minimize the average AoI or to perform remote estimation of a Wiener process.

The works \cite{arxiv02929} and \cite{9834697} considered a system where status updates are transmitted through an unreliable channel and upon delivery of a packet, a feedback message is sent to the sampler, which experiences a random delay. In \cite{arxiv02929}, similarly to \cite{9540757}, a stop-and-wait packet flow is assumed, and the authors derived an optimal sampling policy that minimizes a long-term average age penalty under a sampling rate constraint. In  \cite{9834697}, the authors considered a controller at the sampler that
decides to stay idle or send a fresh update in each slot according to delayed ACK messages.   
They studied two scenarios to derive bounds on the optimal average AoI: 
\begin{enumerate*} \item 
in some specific slots, a genie temporarily takes over the queue in the backward channel and delivers all the queued ACK messages to the sampler instantaneously, and
\item 
the system can contain at most two packets, including both status updates and ACK messages. They showed that the optimal policy in the first scenario provides a lower bound on the optimal average AoI of the  system and the optimal status update policy in the second one provides an upper bound on the optimal average AoI of the  system.
\end{enumerate*}

In contrast to \cite{9540757,arxiv02929}, this paper relaxes the restriction to stop-and-wait mechanisms. As conjectured in \cite[Remark~2]{9540757} and \cite[Footnote~2]{arxiv02929}, this paper shows that, instead of working in a stop-and-wait fashion, requesting status updates \textit{anticipatively} is, in general, beneficial. Indeed, the consideration of pipelining (in the form of two active requests) in a status update system under two-way delay is the most distinctive feature of this paper. Also, different from \cite{9540757,arxiv02929,9834697}, we consider a strictly obedient sampler 
such that once it receives a request message, it is not allowed to wait prior to sampling. Overall, our analysis is quite different and is conducted for a discrete-time system, as opposed to the continuous-time systems in \cite{9540757,arxiv02929}. 
The work in \cite{9834697} is a parallel work to this paper, but in contrast to \cite{9834697},
 we derive closed-form expressions of the average AoI for three different control policies.

\begin{figure}[t]
\centering
\includegraphics[width=.97\linewidth,trim = 0mm 0mm 0mm 0mm,clip]{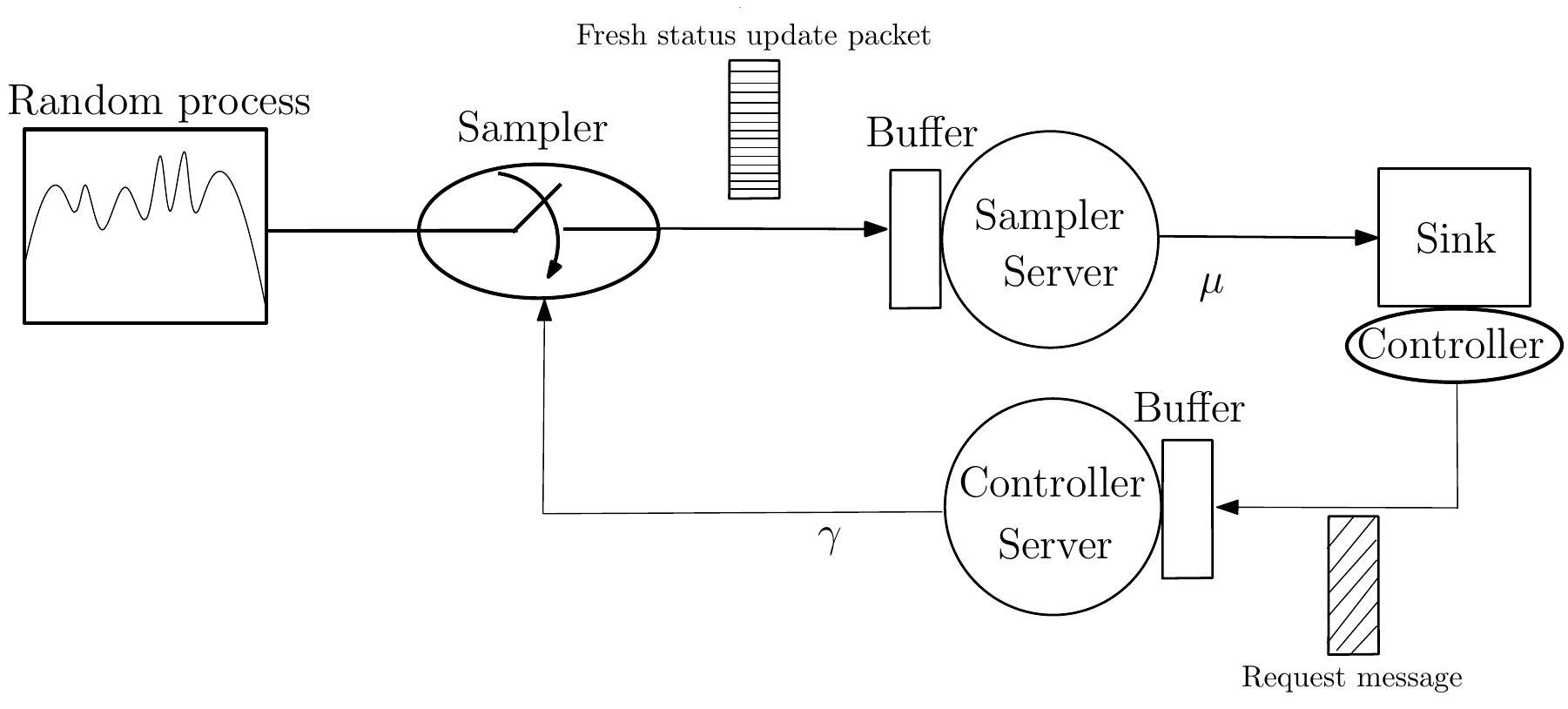}\vspace{-2mm}
\caption{The considered status update system with two-way delay.}
\vspace{-5mm}
\label{Model}
\end{figure}

\section{System Model and AoI Definition}\label{System Model}

\subsection{System Model}
We consider a status update system consisting of one sampler and one sink, as depicted in Fig.~\ref{Model}. We assume slotted communication with normalized slots ${t\in\{0,1,\dots\}}$; slot $t$ indicates the time interval $[t,t+1)$. The sampler monitors a random process and the sink is interested in timely information about the status of the process.

\subsubsection{Request Messaging}
We consider that the sampler can be commanded to take a sample of the process at will. More precisely, we consider a \textit{controller} at the sink who sends \textit{request messages} (i.e., request packets) to the sampler; once the sampler receives the request, it takes a sample. We consider that each request packet requires a random time to be successfully delivered to the sampler. We model this as a \textit{controller server} residing at the link from the controller to the sampler. We assume that the controller server serves request packets according to a geometrically distributed service time with mean
$1/\gamma$. 
 Moreover, the controller server is associated with a  \textit{controller buffer} that stores arriving request packets whenever the server is busy.

\subsubsection{Status Updating}
If the sampler receives a request packet at the end of slot $t-1$, it generates a sample at the beginning of slot $t$. The samples are sent to the sink as \textit{status update packets}, each containing the measured value of the monitored process and a time stamp representing the time when the sample was generated. We consider that each status update packet requires a random time to arrive at the sink. We model this as a server residing at the link from the sampler to the sink, henceforth called the \textit{sampler server}. We assume that the sampler server serves the packets according to a geometrically distributed service time with mean $1/\mu$. The sampler server is associated with a 
\textit{sampler buffer} that stores an arriving status update packet if the server is busy.

\subsubsection{System Capacity}
We will study two scenarios regarding the system capacity, i.e., the maximum number of packets (either request or status update packets) concurrently in the system.  
We define an \textit{active request} as a request that the controller has created, but for which the sink has not yet received the corresponding status update packet. In the first scenario, there is at most one active request concurrently in the system. This corresponds to the typical \textit{stop-and-wait} transmission policy, where a new request packet can be inserted only after the preceding requested status update packet is delivered to the sink. 
To give insights to the queueing dynamics, the system holds at most one status update packet at a time while the buffers are always empty. In the second scenario, we consider a system with at most two active requests. This takes a step toward \textit{pipelining}, where the system is simultaneously processing multiple packets. This leads to more intricate queueing dynamics: for example, the system may have two requests or two status update packets so that one packet is waiting in a buffer.

\subsection{AoI Definition}
The AoI at the sink is defined as the time elapsed since the last successfully received status update packet was generated. Formally, let $t_{i}$ denote the time slot at which the $i$th request message is generated.
This request message completes service at the controller server  at time slot $t'_i-1$  and is delivered to the sampler server at time slot $t'_i$, which is also the time slot that the $i$th status update packet is both generated at the sampler and arrives at the sampler server. 
To close the round trip passage of update $i$, let $\bar t_{i}$ denote the time slot at which the $i$th status update packet arrives at the sink. At time slot $\tau$, the index of the most recently received status update packet is
\begin{equation}\label{mnb00}
N(\tau)=\max\{i\colon \bar t_{i}\le \tau\}, 
\end{equation}
and the time stamp of the most recently received status update packet is
${U^\tau=t'_{N(\tau)}}$.
The (discrete) AoI at the sink is defined as the random process ${\delta^t=t-U^t}$. The evolution of the AoI is given as
\begin{align}\label{eq_age_evo}
&\delta^t=
    \begin{cases}
    t-
t'_{N(t)},& \mbox{if}~{t=\bar t_{N(t)}}
        \\
    \delta^{t-1} + 1,&\mbox{otherwise.}
    \end{cases}
\end{align}
A sample path of the AoI is exemplified in Fig.~\ref{stair_int_1}, which illustrates the staircase pattern of the AoI evolution.

To evaluate the AoI, 
we use the most commonly used metric, the average AoI \cite{8187436,9380899}, which is defined in the following. Let $(0,\tau)$ denote an observation interval. Accordingly, the time average AoI at the sink, denoted as  $\Delta_{\tau}$, is defined as
\begin{equation}\label{oointr}
\Delta_{\tau}=\dfrac{1}{\tau}\textstyle\sum_{t=1}^{\tau}\delta^t.
\end{equation}
The summation in \eqref{oointr} is equal to the area under $\delta^t$, which can be expressed as a sum of disjoint areas determined by polygons $\underline Q$, ${\{Q_{i}\}_{i=2}^{N(\tau)}}$, and $\bar{Q}$, as illustrated in Fig. \ref{stair_int_1}. 
Following the definition of  $N(\tau)$ in \eqref{mnb00}, $\Delta_{\tau}$ can be calculated as
\begin{align}\label{oointr01}
\Delta_{\tau}&=\dfrac{1}{\tau}\bigg(\underline Q+\textstyle\sum_{i=2}^{N(\tau)}Q_{i}+\bar{Q}\bigg)\nn
&=\dfrac{\underline Q+\bar{Q}}{\tau}+\dfrac{N(\tau)-1}{\tau}\dfrac{1}{N(\tau)-1}\sum_{i=2}^{N(\tau)}Q_{i}.
\end{align}

In our analysis, we use a transition between the discrete AoI process defined in \eqref{eq_age_evo} and the AoI process, which is equivalent to \eqref{eq_age_evo} except that the age increases continually during a slot. To this end, let $Q^+_{i}$ denote a trapezoid associated with $Q_{i}$ so that $Q_{i}$ is inscribed in $Q^+_{i}$, as illustrated in Fig. \ref{stair_int_1} for $i=2$. The trapezoids $\underline Q^+$ and $\bar Q^+$ are defined similarly. Thus, we can rewrite $\Delta_{\tau}$ in \eqref{oointr01} as 
\begin{align}\label{oointr02}
\!\!\Delta_{\tau}\!=\!\dfrac{\underline Q^++\bar{Q}^+}{\tau}+\dfrac{N(\tau)-1}{\tau}\dfrac{1}{N(\tau)-1}\textstyle\sum_{i=2}^{N(\tau)}Q^+_{i}-\dfrac{1}{2},
\end{align}
because by considering $\underline Q^+$, $\{Q^+_{i}\}_{i\ge1}$, and $\bar Q^+$, a triangle of area $1/2$ has been added in each slot to the area under $\delta^t$, and thus, we need to subtract the average area of these triangles, which is $1/2$. 

\begin{figure}[t]
\centering
\includegraphics[width=.97\linewidth,trim = 0mm 1mm 0mm 0mm,clip]{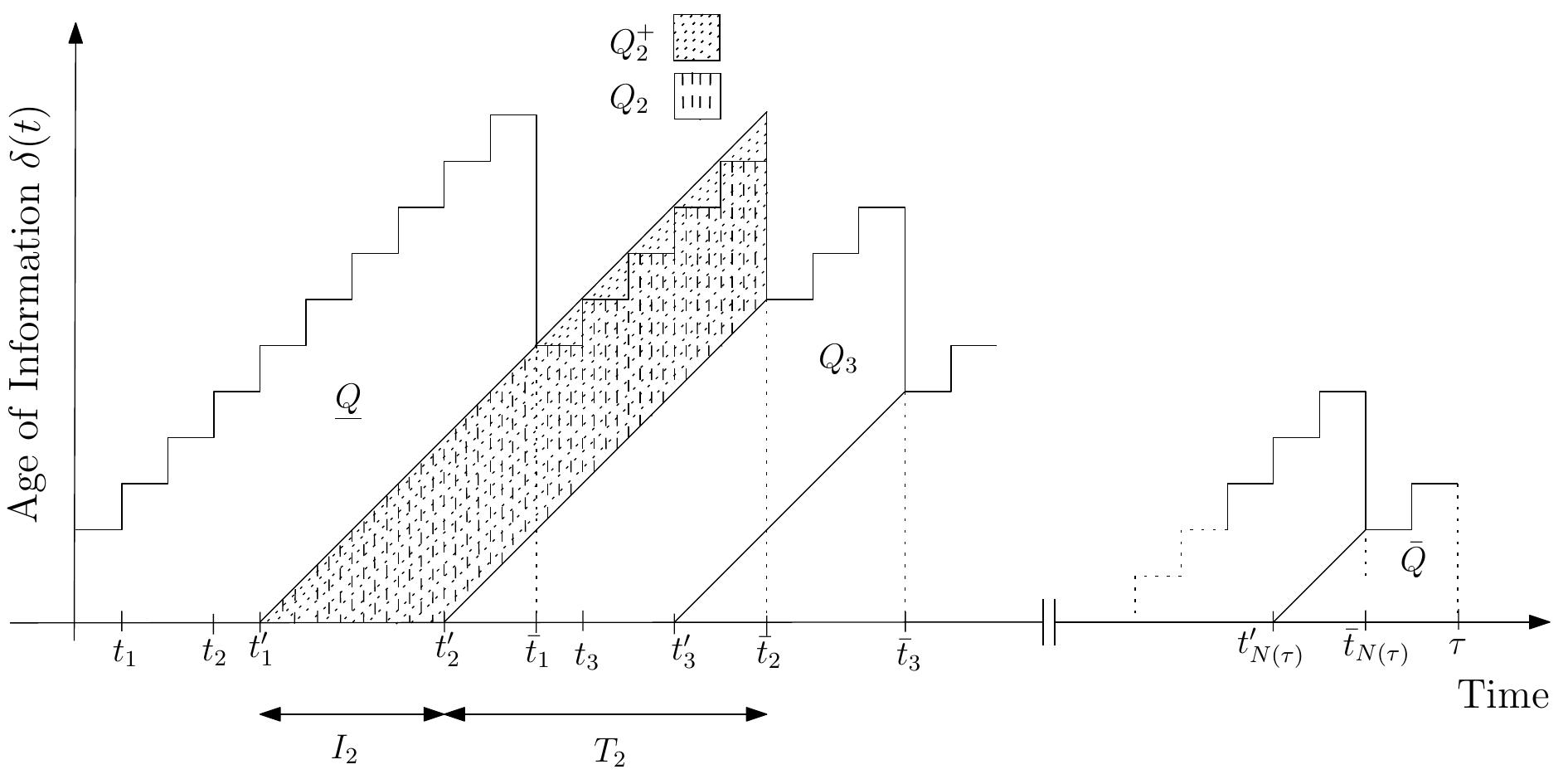}
\caption{An example of the evolution of the AoI. 
} \vspace{-5mm} 
\label{stair_int_1}
\end{figure}

Let $T_i=\bar t_i - t'_i$ denote the system time of status update packet $i$, defined as the time duration between the generation time at the sampler and the reception time at the sink. Let ${I_i=t'_{i}-t'_{i-1}}$ denote the interarrival time between  status update packets $i-1$ and $i$ at the sampler server. 
Since $Q^+_{i}$ can be calculated by subtracting the area of the isosceles triangle with sides $T_i$ from the area of the isosceles triangle with sides ${I_i+T_i}$ (see Fig.~\ref{stair_int_1}), 
\begin{equation}\label{QplusDefn}
Q^+_{i}=\frac{1}{2}I_i^2+I_iT_i.
\end{equation}

The (time) average AoI, denoted by $\Delta$, is defined as
${\Delta=\lim_{\tau\to\infty}\Delta_{\tau}}$.  
With the assumption that $\set{(I_i,T_i)}_{i\ge1}$ is a stationary ergodic random process,\footnote{Such ergodicity assumptions are common in the AoI literature \cite{6195689,8187436}. For the controller policies considered in this work, explicit verification can be nontrivial but not especially interesting.}  the limit $\lambda=\lim_{\tau\to\infty}({N(\tau)-1})/{\tau}$ exists and is given as $\lambda=1/\mathbb{E}[I_i]$ and
the sample average $\frac{1}{N(\tau)-1}\textstyle\sum_{i=2}^{N(\tau)}Q^+_{i}$ in \eqref{oointr02} converges to the stochastic average $\mathbb{E}[Q^+_{i}]$.
%
Since the term $({\underline{Q}^++\bar{Q}^+})/{\tau}$  in \eqref{oointr02} goes to zero as ${\tau\to\infty}$, it then follows from \eqref{oointr02} and \eqref{QplusDefn}
that the average AoI is
\begin{align}\label{A_AoI_Main}
\Delta&=
\dfrac{1}{\mathbb{E}[I_i]}\mathbb{E}\left[\frac{1}{2}I_i^2+I_iT_i\right]-\dfrac{1}{2}.
\end{align}
Going forward, we will use the average AoI  \eqref{A_AoI_Main} to compare controller policies.

\section{Optimal Control for Average AoI 
Minimization
}\label{Optimal Control for AoI Minimization} 
In this section, we study and design optimal control policies for the considered status updating system. More precisely, our goal is to determine the optimal times for the controller to send request messages to minimize the average AoI at the sink. We split the control design into two cases. In Section~\ref{sec_1P_control}, we consider the system with at most one active request, hereinafter referred to as \textit{$\text{1-Packet}$ system}. In Section~\ref{sec_2P_control}, we consider the system with at most two active requests, hereinafter referred to as \textit{$\text{2-Packet}$ system}. 
 
 
In this section, we use a common assumption (see e.g., \cite{9241401,9085402,8938128,8648525}) that the AoI is upper-bounded by a sufficiently large value $\bar\Delta$ so that $\delta^t\in\{1,2,\ldots,\bar\Delta\}$. This enables tractable design of optimal control policies by making the state spaces finite, while acknowledging that when the status information becomes excessively stale by reaching $\bar\Delta$, a time-critical end application would not be affected if counting further. How $\bar\Delta$ is chosen in practice is discussed in the numerical results of Section~\ref{Numerical Results}.

\subsection{$\text{1-Packet}$ System}\label{sec_1P_control}
In this section, we formulate the optimal status update control problem as a Markov decision process (MDP) problem and propose a relative value iteration (RVI) algorithm to find an optimal control policy.

\subsubsection{State} 
Let ${s^t\in\mathcal{S}}$ denote the state of the system in slot $t$, where ${\mathcal{S}}$ is the state space. We define the system state as ${s^t=(\delta^t,e_{\mathrm{s}}^t,E_{\mathrm{s}}^t,\Delta^t_{\mathrm{s}})}$ with the following four elements: 
\begin{enumerate*} \item $\delta^t$ is the AoI at the sink, 
\item $e^t_{\mathrm{s}}\in\{0,1\}$ is an occupancy indicator for the controller server; ${e^t_{\mathrm{s}}=1}$ indicates\footnote{The same convention is used for all binary occupancy indicators throughout the paper.} that there is a packet at the controller server, and ${e^t_{\mathrm{s}}=0}$ otherwise, 
\item $E^t_{\mathrm{s}}\in\{0,1\}$ is an occupancy indicator for the sampler server, and
\item $\Delta^t_{\mathrm{s}}$ is the age of the status update packet at the sampler server. 
\end{enumerate*}

\subsubsection{Action} 
The action taken in slot $t$ is 
${a^t\in\mathcal{A}=\{0,1\}}$, where ${a^t=1}$ indicates that the controller generates a request packet at the beginning of slot $t$, and ${a^t=0}$ otherwise.

\subsubsection{Policy} 
A policy $\pi$ is a (possibly randomized) mapping from the state space $\mathcal{S}$ to action space $\mathcal{A}$. It determines an action $a^t$ taken in the current state $s^t$. 
Note that according to the 1-Packet system constraint, 
an admissible policy $\pi\in\Pi_\textrm{1-Packet}$ must satisfy  ${a^t=0}$ if there is an active request in the system, i.e., if either ${e_\mathrm{s}^t=1}$ or ${E_\mathrm{s}^t=1}$.

\subsubsection{Problem Formulation}
Given an initial state ${s^0\in\mathcal{S}}$, the long-term average AoI, obtained by following a policy $\pi$, is defined as 
\begin{align}
J_{\pi}(s^0)=\limsup_{T\rightarrow\infty}\dfrac{1}{T}\sum_{t=0}^{T-1}\mathbb{E}\left[\delta^t \mid s^0\right].
\end{align}
where the expectation is with respect to the action taken in each slot and the randomness in the service times, i.e., required time to serve each packet in the controller and sampler servers. For the $\text{1-Packet}$ system, the objective is to find the optimal policy $\pi^*$ that provides the minimum long-term average AoI  
\begin{align}
(\textbf{PI})\quad\pi^*(s^0)=\argmin_{\pi \in \Pi_\textrm{1-Packet}} J_{\pi}(s^0),
\end{align}
Next, we recast problem (\textbf{PI}) as an MDP problem and solve it (i.e., find the optimal policy) using the RVI algorithm.

\subsubsection{MDP Modeling of Problem~(\textbf{PI})}\label{MDP Formulation 1}
The MDP of the $\text{1-Packet}$ system is defined by the tuple ${(\mathcal{S},\mathcal{A},\prob{s^{t+1}|s^t,a^t},C(s^t,a^t))}$, 
with $C(s^t,a^t)$ denoting the (immediate) cost of taking action $a^t$ in state $s^t$. Specifically,  we define $C(s^t,a^t)=\delta^{t+1}$, the AoI at slot $t+1$. The state transition probability $\prob{s'|s,a}=\prob{s^{t+1}|s^t,a^t}$ gives the probability of moving from the current state ${s=s^t}$ to the next state ${s'=s^{t+1}}$ under an action ${a=a^t}$. 

To facilitate a compact description of $\prob{s'|s,a}$, we employ the shorthand notations
${\bar\gamma\triangleq1-\gamma}$ and 
${\bar\mu\triangleq1-\mu}$, as well as the operator
\begin{equation}
\Agelimit{\xi}\triangleq\min\{\xi+1,\bar\Delta\}.
\end{equation}
We also adopt the convention that $\star$ indicates an irrelevant age value when the sampler server is idle.
Then, considering the system dynamics, the transition probability from state ${s=(\delta,e_{\mathrm{s}},E_{\mathrm{s}},\Delta_{\mathrm{s}})}$ to the next state ${s'=(\delta',e_{\mathrm{s}}',E_{\mathrm{s}}',\Delta'_{\mathrm{s}})}$ under an action $a$, 
is characterized as:
\begin{subequations}
\begin{align}
&\prob{s'|s=(\delta,0,0,\star),a}\nn
&\quad=
\begin{cases}
1, & a=0,~ s'=(\Agelimit{\delta},0,0,\star); \\
\gamma, & a=1,~s'=(\Agelimit{\delta},0,1,0);\\
\bar\gamma, & a=1,~ s'=(\Agelimit{\delta},1,0,\star);\\
0,& \text{otherwise.}
\end{cases}\\
&\prob{s'|s=(\delta,1,0,\star),a}\nn
&\quad=\begin{cases}
\gamma,   & a=0,~ s'=(\Agelimit{\delta},0,1,0); \\
\bar\gamma, & a=0,~s'=(\Agelimit{\delta},1,0,\star);\\
0,  & \text{otherwise.}
\end{cases}\\
&\prob{s'|s=(\delta,0,1,\Delta_{\mathrm{s}}),a}\nn
&\quad=
\begin{cases}
\mu,   & a=0,~ s'=(\Agelimit{\Delta_{\mathrm{s}}},0,0,\star); \\
\bar\mu, & a=0,~s'=(\Agelimit{\delta},0,1,\Agelimit{\Delta_{\mathrm{s}}});
\\
0,& \text{otherwise.}
\end{cases}
\end{align}
\end{subequations}

Based on the MDP construction above, problem (\textbf{PI}) is equivalently cast as an MDP problem   
\begin{equation}\label{MDP1}
\pi^*(s^0)\!=\!\argmin_{ \pi  \in \Pi_\textrm{1-Packet}} \left\{\limsup_{T\rightarrow\infty}\dfrac{1}{T}\sum_{t=0}^{T-1}\mathbb{E}\left[C(s^t,a^t) \mid s^0\right]\!\right\}. 
\end{equation}

\subsubsection{Optimal Policy}
According to \cite[Proposition 4.2.1]{bertsekas2007dynamic}, for a given initial state $s^0$, if we can find a scalar $\bar V$ and values  $V^*_s$ for all $s\in\mathcal{S}$ that satisfy
\begin{align}\label{BertPro421}
\bar{V}+V^*_s=\min_{a\in\mathcal{A}}\left\{C(s,a)+\sum_{s'\in\mathcal{S}}\Pr(s'|s,a)V^*_{s'}\right\},
\end{align}
then $\bar{V}$ is the optimal average cost for problem (\textbf{PI}),  $J_{\pi^*}(s^0)$, and the actions $a\in\mathcal{A}$ that attain the minimum in \eqref{BertPro421} constitute an optimal policy, $\pi^*(s^0)$. 

Next, we define a condition under which the Bellman's equation in \eqref{BertPro421} has a solution and the optimal average cost is independent of the initial state $s^0$ \cite[Proposition~4.2.3]{bertsekas2007dynamic}.
\begin{definition}[Weak accessibility condition]
The weak accessibility condition holds for an MDP problem if the states can be divided into two subsets $\mathcal{S}_{\mathrm{c}}$ and $\mathcal{S}_{\mathrm{t}}$ such that: I) for any two states $s$ and $s'$ in subset $\mathcal{S}_{\mathrm{c}}$, there exists a stationary policy such that the probability of moving from state $s$ to state $s'$ after some slots is non-zero, and II) all states in $\mathcal{S}_{\mathrm{t}}$ are transient under every stationary policy \cite[Definition~4.2.2]{bertsekas2007dynamic}.
\end{definition}

In the following lemma, we prove that the weak accessibility condition holds for the MDP problem defined in \eqref{MDP1}.

\begin{lemma}\label{W_acc_con_1}
For any $(\mu,\gamma)\in\{(\theta_1,\theta_2) \mid 0<\theta_1,
\theta_2\le1,~(\theta_1,\theta_2)\ne(1,1)\}$, the weak accessibility condition holds for the MDP problem \eqref{MDP1}.
\end{lemma}

\begin{proof}
If all  {deterministic}  policies in $\Pi_\textrm{1-Packet}$ are unichain, the weak accessibility condition holds \cite[Proposition~4.2.5]{bertsekas2007dynamic}. A  policy is unichain if in the Markov chain induced by the policy, there exists a state that can be reached from any other states with a non-zero probability 
\cite[Exercise~4.3]{gallager2013stochastic}. 
In the MDP problem \eqref{MDP1}, under any deterministic policy: I) for ${\gamma\ne1}$ and ${0<\mu\le1}$, the state $(\bar\Delta,1,0,\star)$ is accessible from any other states and
II) for ${\gamma=1}$ and ${0<\mu<1}$, the state $(\bar\Delta,0,1,\bar\Delta)$ is accessible from any other states. The justification for case I is that when there is a packet at the controller server, the probability of the event of having unsuccessful transmissions in $\bar{\Delta}$ consecutive slots is non-zero.  A similar argument can be stated for case II. 
Thus, every deterministic policy is unichain and consequently, the weak accessibility condition holds.
\end{proof}
Given that the weak accessibility condition holds, we can find an optimal deterministic policy $\pi^*$ and the optimal average cost $J^*$, which are independent of the initial state $s^0$, by solving the Bellman's equation in \eqref{BertPro421} \cite[Proposition~4.2.6]{bertsekas2007dynamic}. We solve it using the RVI algorithm \cite[Section~4.3]{bertsekas2007dynamic}, which turns the Bellman's equation into an iterative procedure. The steps of the RVI algorithm are summarized in Algorithm~\ref{RVI_algo}. 
As specified in lines~\ref{Algoline-pi} and ~\ref{Algoline-V}, the optimal deterministic action  $\pi^*_s$ and the scalar $V^*_s$ are updated for each state ${s\in\mathcal{S}}$ in each iteration, where ${s_{\text{ref}}\in\mathcal{S}}$ is an arbitrary chosen reference state.
 Once the iterative process has converged, the algorithm provides an optimal policy $\pi^*$ and the optimal value of the average AoI, obtained as $J_{\pi^*}=V^*_{s_{\text{ref}}}$.
The following theorem shows that the RVI algorithm presented in Algorithm~\ref{RVI_algo} converges to an optimal deterministic policy.

\begin{algorithm}[t]
\SetAlgoLined
\textbf{Input:} 1) State transition probabilities $\Pr(s'|s,a)$, 2) Stopping criterion threshold $\epsilon$;

\textbf{Initialization:} 
1) Set ${V^*_s=0}$, ${\forall s\in\mathcal{S}}$, 2) determine an arbitrary reference state ${s_{\text{ref}}\in\mathcal{S}}$, and 3) set $\phi>\epsilon$\;  

\While{$
\phi>\epsilon$}{
\For  {$s\in\mathcal{S}$}{
$\pi^*_s\leftarrow\arg{\displaystyle \min_{a\in\mathcal{A}}}\{C(s,a)+\sum_{s'\in\mathcal{S}}\Pr(s'|s,a)V^*_{s'}\}$\;\label{Algoline-pi}
 $V_s\leftarrow
C(s,\pi^*_s)+\sum_{s'\in\mathcal{S}}\Pr(s'|s,\pi^*_s)V^*_{s'} \label{Algoline-V}
-V^*_{s_{\text{ref}}}$\;

}
$\phi\leftarrow\max_{s\in\mathcal{S}}\{|V_s-V^*_s|\}$\;
$V^*_s\leftarrow V_s$, $\forall s\in\mathcal{S}$\;
}
\textbf{Output:} {{1) Optimal policy 
$\pi^*$,
2) optimal objective function value ${J_{\pi^*}=V^*_{s_{\text{ref}}}}$}.
}
\caption{ RVI algorithm }
\label{RVI_algo}
\end{algorithm} 

\begin{theorem}\label{RVI_conver}
The RVI algorithm presented in Algorithm~\ref{RVI_algo} converges to an optimal deterministic policy and returns the optimal value of the average AoI.
\end{theorem}
\begin{proof}
According to \cite[Page~209]{bertsekas2007dynamic}, it is sufficient to show that the Markov chain induced by every deterministic  policy is unichain and aperiodic. In the proof of Lemma~\ref{W_acc_con_1}, we showed that any deterministic policy induces a unichain Markov chain. According to   \cite[Exercise~4.1]{Gallager1996}, since the recurrent states  $(\bar\Delta,1,0,\star)$ (for ${\gamma\ne1}$ and ${0<\mu\le1}$)
and  $(\bar\Delta,0,1,\bar\Delta)$ (for ${\gamma=1}$ and ${0<\mu<1}$) have self-transitions, the Markov chain induced by every deterministic  policy is aperiodic.
\end{proof}

\subsection{$\text{2-Packet}$ System}\label{sec_2P_control}
Here, we extend the MDP construction of the $\text{1-Packet}$ system (described in Section \ref{sec_1P_control}) for the $\text{2-Packet}$ system and subsequently derive the optimal status update control policy.

For the $\text{2-Packet}$ system, the action space and the cost function are the same as for the $\text{1-Packet}$ system. However, the system state needs to be redefined to account for the two active requests.

\subsubsection{State}
Let ${\underline s^t\in\mathcal{\underline S}}$ denote the system state in slot $t$, where ${\mathcal{\underline S}}$ is the state space. We define the system state as ${\underline s^t=(\delta^t,e_{\mathrm{b}}^t,e_{\mathrm{s}}^t,E_{\mathrm{b}}^t,E_{\mathrm{s}}^t,\Delta^t_{\mathrm{b}},\Delta^t_{\mathrm{s}})}$ with the following seven elements: 
\begin{enumerate*}
\item $\delta^t$ is the AoI at the sink, 
\item ${e^t_{\mathrm{b}}\in\{0,1\}}$ is the occupancy indicator for the controller buffer,
\item ${e^t_{\mathrm{s}}\in\{0,1\}}$ is the occupancy indicator for the controller server,
\item ${E^t_{\mathrm{b}}\in\{0,1\}}$ is the occupancy indicator for the sampler buffer,
\item ${E^t_{\mathrm{s}}\in\{0,1\}}$ is the occupancy indicator for the sampler server, 
\item $\Delta^t_{\mathrm{b}}$ is the age of the status update packet at the sampler buffer, and 
\item $\Delta^t_{\mathrm{s}}$ is the age of the status update packet at the sampler server. 
\end{enumerate*}

\begin{remark}\label{remark_state_aware}
Having the four occupancy indicators $e^t_{\mathrm{b}}$, $e^t_{\mathrm{s}}$,  $E^t_{\mathrm{b}}$, and $E^t_{\mathrm{s}}$ in the state implies that the controller (i.e., the decision-maker) has global knowledge of the system occupancy. This consideration may not always be practical, and thus, the derived control policy establishes a performance bound for policies operating under local/imperfect knowledge.
\end{remark}

\subsubsection{Problem Formulation}
Given an initial state $\underline s^0$, the average AoI, obtained by following a policy $\pi$, is defined as 
\begin{align}
\underline J_{\pi}(\underline s^0)=\limsup_{T\rightarrow\infty}\dfrac{1}{T}\sum_{t=0}^{T-1}\mathbb{E}_{\pi}\left[\delta^t\mid\underline s^0\right].
\end{align}
Note that according to the 2-Packet system constraint, the set of admissible policies must satisfy that if there are two active requests in the system, i.e., ${(e_\mathrm{b}^t,e_\mathrm{s}^t)=(1,1)}$ or ${(e_\mathrm{s}^t,E_\mathrm{s}^t)=(1,1)}$ or ${(E_\mathrm{b}^t,E_\mathrm{s}^t)=(1,1)}$, then ${a^t=0}$. The problem of finding the optimal policy $\pi^*$ that provides the minimum long-term average AoI in the $\text{2-Packet}$ system is formulated as 
\begin{align}
(\textbf{PII})~\pi^*(\underline s^0)=\argmin_{\pi \in \Pi_\textrm{2-Packet}} \underline J_{\pi}(\underline s^0),
\end{align}
where $\Pi_\textrm{2-Packet}$ is the set of all admissible policies subject to the 2-Packet system constraint. Next, we recast problem (\textbf{PII}) as an MDP problem and solve it using the RVI algorithm.

\subsubsection{MDP Modeling for Problem (\textbf{PII})}\label{MDP Formulation 2}
The MDP of the $\text{2-Packet}$ system is defined by the tuple $(\mathcal{\underline S},\mathcal{A},\Pr(\underline s^{t+1}|\underline s^t,a^t),C(\underline s^t,a^t))$. The (immediate) cost function is defined as $C(\underline s^t,a^t)=\delta^{t+1}$ (similarly as for problem (\textbf{PI})). The transition probability $\Pr(\underline s'|\underline s,a)$ of moving from state 
$\underline s=(\delta,e_{\mathrm{b}},e_{\mathrm{s}},E_{\mathrm{b}},E_{\mathrm{s}},\Delta_{\mathrm{b}},\Delta_{\mathrm{s}})$
to the next state 
$\underline s'=(\delta',e_{\mathrm{b}}',e_{\mathrm{s}}',E_{\mathrm{b}}',E_{\mathrm{s}}',\Delta_{\mathrm{b}}',\Delta_{\mathrm{s}}')$ 
under an action $a$ is characterized as: 
\begin{subequations}
\begin{IEEEeqnarray}{lCl}
\mathrlap{\Pr(\underline s'|\underline s=(\delta,0,0,0,0,\star,\star),a)}&&\nn
&=&\begin{cases}
1, & a=0,~ \underline s'=(\Agelimit{\delta},0,0,0,0,\star,\star); \\
\gamma, & a=1,~\underline s'=(\Agelimit{\delta},0,0,0,1,\star,0);\\
\bar\gamma, & a=1,~ \underline s'=(\Agelimit{\delta},0,1,0,0,\star,\star);\\
0,& \text{otherwise.}
\end{cases}\\
\mathrlap{\Pr(\underline s'|\underline s=(\delta,1,1,0,0,\star,\star),a)}&&\nn
&=&\begin{cases}
\gamma, & a=0,~\underline s'=(\Agelimit{\delta},0,1,0,1,\star,0);\\
\bar\gamma, & a=0,~ \underline s'=(\Agelimit{\delta},1,1,0,0,\star,\star);\\
0,& \text{otherwise.}
\end{cases}\\
\mathrlap{\Pr(\underline s'|\underline s=(\delta,0,1,0,0,\star,\star),a)}&&\nn
&=&\begin{cases}
\gamma, & a=0,~\underline s'=(\Agelimit{\delta},0,0,0,1,\star,0);\\
\bar\gamma, & a=0,~ \underline s'=(\Agelimit{\delta},0,1,0,0,\star,\star);\\
\gamma, & a=1,~ \underline s'=(\Agelimit{\delta},0,1,0,1,\star,0);\\
\bar\gamma, & a=1,~ \underline s'=(\Agelimit{\delta},1,1,0,0,\star,\star);\\
0,& \text{otherwise.}
\end{cases}\\
\mathrlap{\Pr(\underline s'|\underline s=(\delta,0,1,0,1,\star,\Delta_{\mathrm{s}}),a)}&&\nn
&=&\begin{cases}
\bar\gamma\bar\mu, & a=0,~ \underline s'=(\Agelimit{\delta},0,1,0,1,\star,\Agelimit{\Delta_{\mathrm{s}}}); \\
\gamma\bar\mu,     & a=0,~ \underline s'=(\Agelimit{\delta},0,0,1,1,0,\Agelimit{\Delta_{\mathrm{s}}});\\
\bar\gamma\mu,     &a=0,~ \underline s'= (\Agelimit{\Delta_{\mathrm{s}}},0,1,0,0,\star,\star);\\
\gamma\mu,         &a=0,~ \underline s'=(\Agelimit{\Delta_{\mathrm{s}}},0,0,0,1,\star,0);\\
0,& \text{otherwise.}
\end{cases}\\
\mathrlap{\Pr(\underline s'|\underline s=(\delta,0,0,1,1,\Delta_{\mathrm{b}},\Delta_{\mathrm{s}}),a)}&&\nn
&=&\begin{cases}
\mu, &   a=0,~ \underline s'= (\Agelimit{\Delta_{\mathrm{s}}},0,0,0,1,\star,\Agelimit{\Delta_{\mathrm{b}}});\\
\bar\mu, & a=0,~ \underline s'=(\Agelimit{\delta},0,0,1,1,\Agelimit{\Delta_{\mathrm{b}}},\Agelimit{\Delta_{\mathrm{s}}});\\
0,& \text{otherwise.}
\end{cases}\IEEEeqnarraynumspace\\
\mathrlap{\Pr(\underline s'|\underline s=(\delta,0,0,0,1,\star,\Delta_{\mathrm{s}}),a)}&&\nn
\quad&=&\begin{cases}
\bar\mu, & a=0,~\underline s'=(\Agelimit{\delta},0,0,0,1,\star,\Agelimit{\Delta_{\mathrm{s}}}); \\
\mu, & a=0,~\underline s'= (\Agelimit{\Delta_{\mathrm{s}}},0,0,0,0,\star,\star);\\
\bar\gamma\bar\mu, & a=1,~\underline s'=(\Agelimit{\delta},0,1,0,1,\star,\Agelimit{\Delta_{\mathrm{s}}}); \\
\bar\gamma\mu, & a=1,~\underline s'=(\Agelimit{\Delta_{\mathrm{s}}},0,1,0,0,\star,\star);\\
\gamma\bar\mu, & a=1,~\underline s'=(\Agelimit{\delta},0,0,1,1,0,\Agelimit{\Delta_{\mathrm{s}}}); \\
\gamma\mu, & a=1,~\underline s'=(\Agelimit{\Delta_{\mathrm{s}}},0,0,0,1,\star,0);\\
0,& \text{otherwise}.
\end{cases}
\end{IEEEeqnarray}
\end{subequations}
 
Based on the above definitions, the MDP formulation of problem (\textbf{PII}) is formulated as
\begin{align}\label{MDP2}
\!\!\pi^*(\underline s^0)\!=\!\argmin_{ \pi  \in \Pi_\textrm{2-Packet}} \Bigl\{\limsup_{T\rightarrow\infty}\dfrac{1}{T}\sum_{t=0}^{T-1}\mathbb{E}\left[C(\underline s^t,a^t) \mid \underline s^0\right]\Bigr\}. 
\end{align}

\subsubsection{Optimal Policy}
First, similar to the proof of Lemma~\ref{W_acc_con_1} for the $\text{1-Packet}$ system, one can show that the weak accessibility condition holds for the $\text{2-Packet}$ MDP problem in \eqref{MDP2} for any $(\mu,\gamma)\in\{(\theta_1,\theta_2) \mid 0<\theta_1,\theta_2\le1,~(\theta_1,\theta_2)\ne(1,1)\}$. Here, the recurrent states are 
$(\bar\Delta,1,1,0,0,\star,\star)$ for ${\gamma\ne1}$ and ${0<\mu\le1}$, and 
$(\bar\Delta,0,0,1,1,\bar\Delta,\bar\Delta)$ for ${\gamma=1}$ and ${0<\mu<1}$.

By the weak accessibility condition, the optimal average cost $\underline J_{\pi^*}$ and optimal policy $\pi^*$ are derived by solving the following equation for each $\underline s\in\mathcal{\underline S}$ (cf. \eqref{BertPro421}):
\begin{align}\label{B_Eq_2}
\underline J_{\pi^*}+V^*_{\underline s}=\min_{a\in\mathcal{A}}\Bigl\{C(\underline s,a)+\sum_{\underline s'\in\mathcal{\underline S}}\Pr(\underline s'|\underline s,a)V^*_{\underline s'}\Bigr\}.
\end{align}
The RVI algorithm presented in Algorithm \ref{RVI_algo} is used to iteratively solve for $\underline J_{\pi^*}$ and $V^*_{\underline s}$ in \eqref{B_Eq_2}.

\section{Average AoI Analysis}\label{Average AoI Analysis}
In this section, we 
derive the average AoI expressions for the following three policies: \begin{enumerate*} \item a zero-wait policy for the 1-Packet system, hereinafter referred to as
the {\em \ZW{1}} policy,
\item a zero-wait policy for the 2-Packet system, hereinafter referred to as the \emph{\ZW{2}} policy
(Section \ref{sec_2P_ZW}), and \item a waiting-based policy for the 1-Packet system, hereinafter referred to as the \emph{\Wone{}} policy (Section \ref{sec_1P_Wait}).
\end{enumerate*}

\subsection{$\ZWone$: AoI Analysis}\label{sec_1P_ZW}
In this section, we derive the average AoI under the $\ZWone$ policy. Since the $\ZWone$ policy conforms to the \textit{zero-wait} serving discipline (also known as \textit{work-conserving} or \textit{just-in-time} policy \cite{8000687}), the controller sends a new request message immediately after receiving a status update packet. That is, if status update packet $i-1$ is received by the end of slot $t-1$, the controller generates request packet $i$ at the beginning of slot $t$, which enters an empty system. The AoI evolution under $\ZWone$ is illustrated in Fig.~\ref{Zero-W-1}.

To characterize the average AoI as given by \eqref{A_AoI_Main}, we need to calculate $\mathbb{E}[I_i]$, $\mathbb{E}[I_i^2]$, and $\mathbb{E}[I_iT_i]$.
As preliminaries, let $X_i$ be the random variable representing the required time to serve the $i$th request packet by the controller server. Similarly, let $Y_i$ be the random variable representing the required time to serve the $i$th status update packet by the sampler server. 
Because the service times of all request messages are stochastically identical 
and the service times of all status update packets are stochastically identical, ${X_i=^{\mathrm{st}}X}$ and ${Y_i=^{\mathrm{st}}Y}$ for all ${i=\{1,2,\ldots\}}$, and ${X\sim\mathrm{Geo}(\gamma)}$ and ${Y\sim\mathrm{Geo}(\mu)}$.  Now we delve into calculations of $\mathbb{E}[I_i]$, $\mathbb{E}[I_i^2]$, and $\mathbb{E}[I_iT_i]$.

\subsubsection{Derivation of $\mathbb{E}[I_i]$ and $\mathbb{E}[I^2_i]$}
As it can be seen in Fig.~\ref{Zero-W-1}, the 
interarrival time  between status update packets $i-1$ and $i$ at the sampler server is 
\begin{equation}
 I_i=Y_{i-1} + X_i. 
\end{equation} 
The mean and second moment of the interarrival time, $\mathbb{E}[I_i]$ and $\mathbb{E}[I_i^2]$, are 
\begin{align}
\mathbb{E}[I_i]&=\mathbb{E}[Y_{i-1}]+\mathbb{E}[X_{i}]=\dfrac{1}{\mu}+\dfrac{1}{\gamma},\label{I_Z_1}\\
\mathbb{E}[I^2_i]&=\mathbb{E}[(Y_{i-1}+X_{i})^2]\nn
&=\mathbb{E}[Y^2_{i-1}]+\mathbb{E}[X^2_{i}]+2\mathbb{E}[Y_{i-1}]\mathbb{E}[X_{i}]\label{I-2_Z_1a}\\
&=\dfrac{2-\mu}{\mu^2}+\dfrac{2-\gamma}{\gamma^2}+\dfrac{2}{\mu\gamma},
\label{I-2_Z_1}
\end{align}
where \eqref{I-2_Z_1a}
follows from independence 
of $X_{i}$ and $Y_{i-1}$. 

\subsubsection{Derivation of $\mathbb{E}[I_iT_i]$}
As it can be seen in Fig.~\ref{Zero-W-1}, the system time of packet $i$, $T_i$, equals $Y_i$, the service time of update packet $i$. Thus, 
\begin{align}
\mathbb{E}[I_iT_i]
&=\mathbb{E}[(Y_{i-1}+X_{i})Y_{i}]\nn
&\overset{(a)}=\mathbb{E}[Y_{i}]\mathbb{E}[Y_{i-1}+X_{i}]=\dfrac{1}{\mu}\left(\dfrac{1}{\mu}+\dfrac{1}{\gamma}\right)\label{C_Z_1},
\end{align}
where $(a)$ follows from independence of $Y_{i}$ and  $X_i,Y_{i-1}$.

\begin{figure}[t]
\centering
\includegraphics[width=.97\linewidth,trim = 0mm 1mm 0mm 0mm,clip]{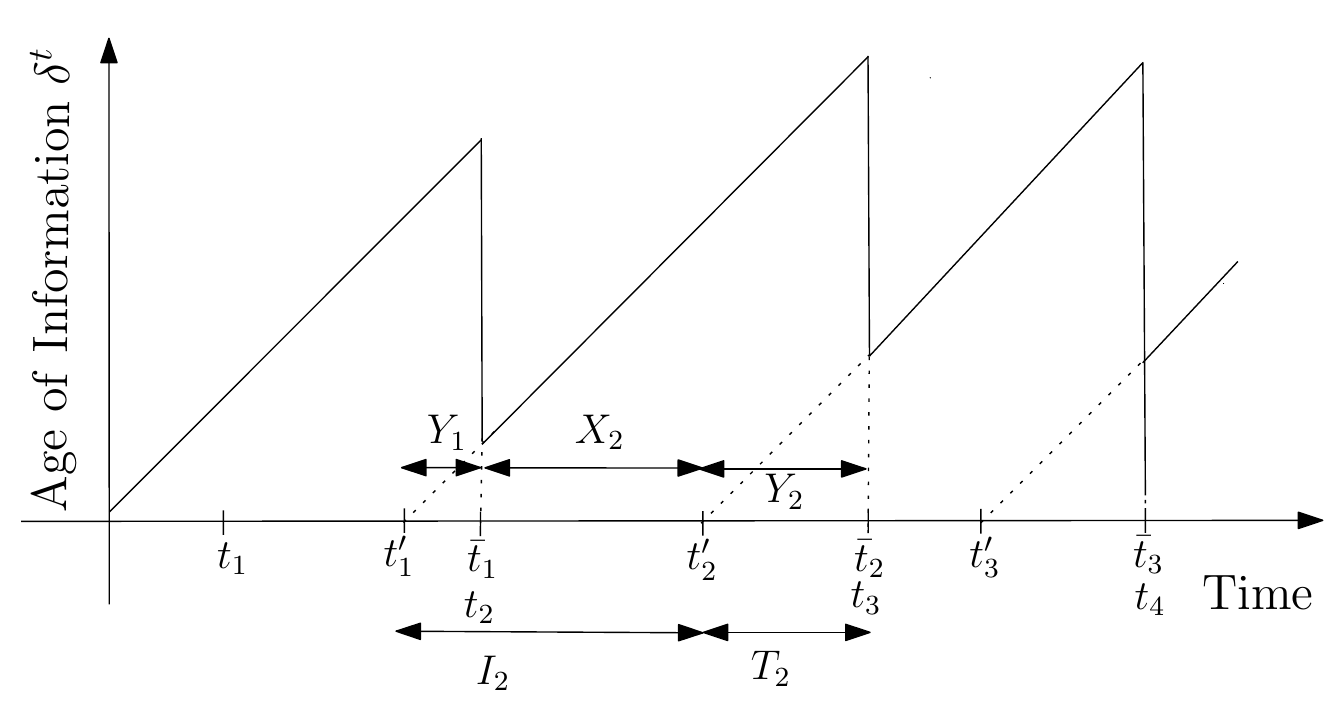}\vspace{-2mm}
\caption{AoI as a function of time under the $\ZWone$ policy.}  
\label{Zero-W-1}\vspace{-5mm}
\end{figure}

Substituting \eqref{I_Z_1}, \eqref{I-2_Z_1}, and \eqref{C_Z_1} into \eqref{A_AoI_Main} yields Theorem~\ref{Age_Zero-W-1}.  
\begin{theorem}\label{Age_Zero-W-1}
The average AoI under the $\ZWone$ policy is
\begin{align}\label{ZW1exp}
\Delta^{\mathrm{\ZWone}}
=\frac{2}{\mu}+\frac{\mu}{\gamma(\mu+\gamma)}-1.
\end{align}
\end{theorem}

\subsection{\ZW{2}: AoI Analysis}\label{sec_2P_ZW}
To evaluate the AoI of the $\ZW{2}$ system using  \eqref{A_AoI_Main}, we need to track the event $B_{i}$ that the sampler server is busy (serving update packet $i-1$)  when request message $i$ arrives at the sampler and update packet $i$ is generated. In particular, with $\Bbar_i$ denoting the complement of $B_i$,  the two-step partition 
$\set{\Bbar_{i-1}\Bbar_{i},\Bbar_{i-1}B_{i},B_{i-1}\Bbar_{i},B_{i-1}B_{i}}$ is needed to evaluate the AoI of the \ZW{2} system. The proof in Appendix~\ref{ZW2-proof} verifies the following result.
\begin{theorem}\label{thm:ZW2} The average AoI under the $\ZWtwo$ policy is
\begin{align}
 \Delta^{\ZWtwo}   &=\frac{1}{\gamma}+\frac{1}{\mu}-1
    +\frac{2\gamma^2\bar{\mu}}{\mu(\gamma\bar{\mu}(\gamma+\mu)+\mu^2)}.
\end{align}
\end{theorem}

By applying algebraic manipulations to the difference ${\Delta^{\mathrm{\ZWtwo}}-\Delta^{\mathrm{\ZWone}}}$, we obtain the following corollary that specifies the ranges of service time parameters $\gamma$ and $\mu$ for which $\ZWtwo$ outperforms the $\ZWone$. 
\begin{corollary}\label{Z_W_2&Z_W_1}
The $\ZWtwo$ policy outperforms the $\ZWone$ policy, i.e.,  ${\Delta^{\mathrm{\ZWtwo}}\le\Delta^{\mathrm{\ZWone}}}$, if the service rates satisfy
\begin{subequations}
\begin{align}
\label{Compa_W1_W2}
\mu &\ge\frac{\gamma\big(\bar\gamma+ \sqrt{(\gamma + 1)^2 + 4}\big)}{2(\gamma + 1)}\quad \text{or} \\
\gamma &\le \frac{\mu\big(-\bar\mu+ \sqrt{\bar\mu(5-\mu)}\big)}{2+\bar\mu}.
\end{align}
\end{subequations}
\end{corollary}
The  specified regions 
are illustrated in Fig. \ref{ZW_1P_2P_region}.
As one observation from \eqref{Compa_W1_W2}, when $\mu\ge \frac{\sqrt{2}}{2}\approx0.7071$, we have $\Delta^{\mathrm{\ZWtwo}}\le\Delta^{\mathrm{\ZWone}}$ for any value of $\gamma$.

\begin{figure}[t]
\centering
\includegraphics[width=.9\linewidth,trim = 0mm 0mm 0mm 0mm,clip]{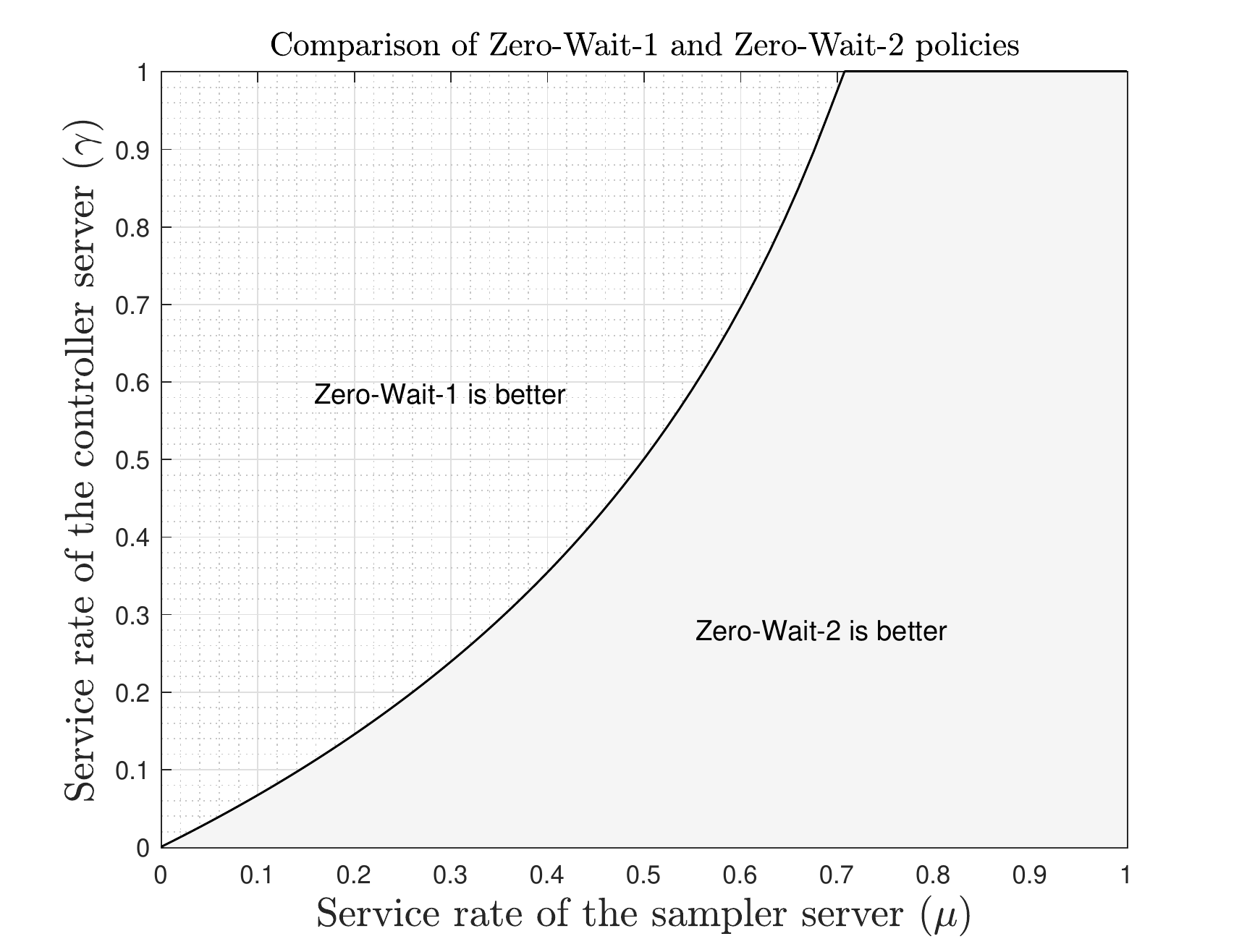}
\caption{Average AoI performance comparison between the $\ZWone$ policy and the $\ZWtwo$ policy with respect to different service rate pairs $(\mu,\gamma)$.}\vspace{-2mm}  
\label{ZW_1P_2P_region}\vspace{-3mm} 
\end{figure}

\subsection{$\Wone$: AoI Analysis}\label{sec_1P_Wait}
In this section, we derive the average AoI under the $\Wone$ policy. Differently to the zero-wait policies, when the controller receives status update packet ${i-1}$, the $\Wone$ policy permits the controller to  assign a \textit{waiting time} prior to sending request message $i$. Such appropriately placed waiting time has potential to lead to lower average AoI than that of the zero-wait policy\footnote{This is not unforeseen: a waiting-based policy, instead of a zero-wait policy, has been shown to realize optimal status updating in various other systems, as detailed in Remark~\ref{remark_beta_policy}.}, 
as will be shown by our subsequent analysis. 

Formally, let ${Z_i}$ be a random variable for the waiting time assigned to request message $i$ after receiving update ${i-1}$. In this work, in line with \cite{7283009}, we consider that ${Z_i=f(Y_{i-1})}$, i.e., the waiting time is a deterministic function of the service time of status update packet $i-1$ (at the sampler server). AoI evolution under the $\Wone$ policy is illustrated in Fig.~\ref{Wait-1}. 

\begin{figure}[t]
\centering
\includegraphics[width=.97\linewidth,trim = 0mm 1mm 0mm 0mm,clip]{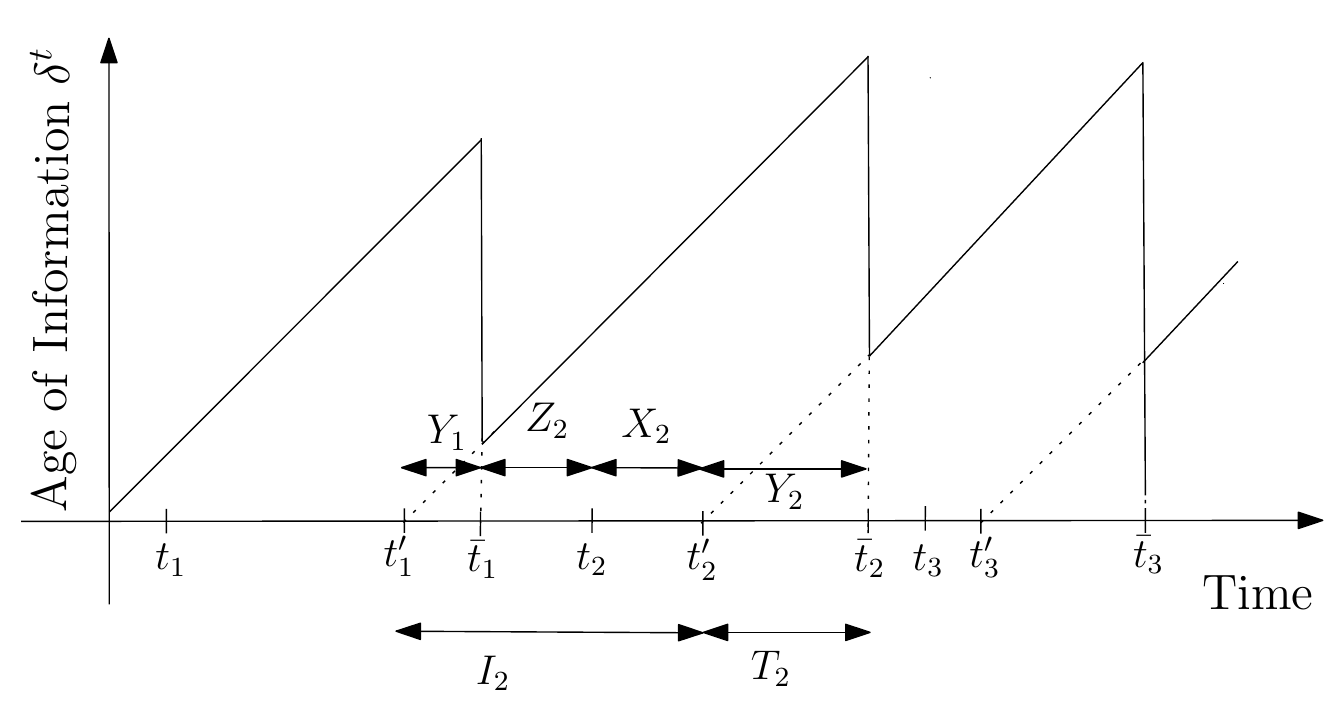}\vspace{-2mm}
\caption{AoI as a function of time under the $\Wone$ policy.}
\vspace{-5mm}
\label{Wait-1}
\end{figure}

We restrict to  \textit{threshold-type} waiting functions of form 
\begin{align}\label{threshold-f}
Z_i=f(Y_{i-1})=(\beta-Y_{i-1})^+,
\end{align} 
where $\beta$ is a positive integer;  ${\beta=1}$ corresponds to the $\ZWone$ policy. While we have not been able to prove the optimality of this waiting function for our considered system, the structure of an optimal policy (obtained by the MDP approach), as visualized in Section \ref{Numerical Results}, has such threshold-type structure.
   
\begin{remark}\label{remark_beta_policy}
This is consistent with prior works on the optimality of waiting policies for various  status updating systems. The threshold-type waiting function
\eqref{threshold-f} with appropriately adjusted $\beta$, was  proven to be the optimal waiting function in a stop-and-wait continuous-time system with a non-delayed control channel in \cite[Theorem 1]{7283009} and \cite[Theorem 4]{8000687} and in its discrete-time equivalent in \cite[Theorem~3]{8764465}. Under a delayed control channel, the optimality of the threshold-type waiting in continuous-time systems has been shown in \cite[Proposition 2]{9540757} and \cite[Theorem 1]{arxiv02929}. 
\end{remark}

To characterize the average AoI as given by \eqref{A_AoI_Main}, we need to calculate $\mathbb{E}[I_i]$, $\mathbb{E}[I_i^2]$, and $\mathbb{E}[I_iT_i]$. As it can be seen in Fig.~\ref{Wait-1}, the interarrival time of status update packets $i-1$ and $i$ is the summation of the service time of status update packet $i-1$, $Y_{i-1}$, the waiting time of request message $i$, $Z_i$, and the service time of request message $i$, $X_{i}$, i.e., ${I_i=Y_{i-1} + Z_i + X_i}$. The mean interarrival time, $\mathbb{E}[I_i]$, is given as
\begin{align}\label{I_Z_1_Wait1}
\mathbb{E}[I_i]=\mathbb{E}[Y_{i-1}]+\mathbb{E}[Z_{i}]+\mathbb{E}[X_{i}]=\dfrac{1}{\mu}+\mathbb{E}[Z_{i}]+\dfrac{1}{\gamma},
\end{align}
where $\mathbb{E}\left[Z_i \right]$ is given by the following lemma.
\begin{lemma}\label{lemma_EZ}
The mean waiting time $\mathbb{E}\left[Z_i \right]$ is given as
\begin{equation}
\begin{array}{ll}
\mathbb{E}\left[Z_i \right] = \displaystyle\frac{1}{\mu}\left( \beta\mu + \bar\mu^\beta -1 \right).
\end{array}
\end{equation}
\end{lemma}
\begin{proof} 
See Appendix~\ref{appendix_EZ}.
\end{proof}

The second moment of the interarrival time, $\mathbb{E}[I_i^2]$, is derived as
\begin{align}\label{I-2_Z_1_Wait1}
\mathbb{E}[I^2_i]&=\mathbb{E}[(Y_{i-1} + Z_i + X_i)^2]\nn
&=\mathbb{E}[(Y_{i-1} + X_i)^2] + 2\mathbb{E}[Z_i(Y_{i-1} + X_i)] + \mathbb{E}[ Z_i^2 ]\nn
&\overset{(a)}{=}\mathbb{E}[Y^2] + \mathbb{E}[X^2] + 2\mathbb{E}[X]\mathbb{E}[Y] + 2\mathbb{E}[Z_iY_{i-1}] \nn
&\qquad+ 2\mathbb{E}[Z_i]\mathbb{E}[X_i]  + \mathbb{E}[ Z_i^2 ],
\end{align}
where equality $(a)$ comes from the fact that $X_{i}$ and $Y_{i-1}$ are independent and $Z_{i}$ and $X_{i}$ are independent; the expectations $\mathbb{E}[Z_iY_{i-1}]$ and $\mathbb{E}\left[Z_i^2 \right]$ are given by the following lemmas.
\begin{lemma}\label{lemma_EZY}
The expectation $\mathbb{E}\left[Z_iY_{i-1} \right]$ is given as
\begin{align}
\mathbb{E}\left[Z_iY_{i-1} \right]
&=\frac{1}{\mu^2}\big( \mu + \beta\mu +\bar\mu^\beta(2-\mu+\beta\mu) - 2\big)
\end{align}
\end{lemma}
\begin{proof} 
See Appendix~\ref{appendix_EZY}.
\end{proof}
\begin{lemma}\label{lemma_EZ2}
The second moment of the waiting time, $\mathbb{E}\left[Z_i^2 \right]$, is given as
\begin{equation}
\begin{array}{ll}
\mathbb{E}\left[Z_i^2 \right]&\hspace{-2mm} \!=\! \displaystyle-\frac{1}{\mu^2}\big(\mu + 2\beta\mu + \bar\mu^\beta(2-\mu)- \beta^2\mu^2 - 2\big).
\end{array}
\end{equation}
\end{lemma}
\begin{proof} 
See Appendix~\ref{appendix_EZ2}.
\end{proof}

As it can be seen in Fig.~\ref{Wait-1}, the system time of packet $i$, $T_i$, is equal to the service time of status update packet $i$, $Y_i$. Thus, the expectation $\mathbb{E}[I_iT_i]$ is derived as
\begin{align}
\mathbb{E}[I_iT_i]&=\mathbb{E}[(Y_{i-1}+Z_{i}+X_{i})Y_{i}]\nn
 &\overset{(a)}{=} \mathbb{E}[Y_{i}]\mathbb{E}[Y_{i-1}+Z_{i}+X_{i}]\nn
 &=\!\dfrac{1}{\mu}\!\left(\dfrac{1}{\mu}+\mathbb{E}[Z_i]+\dfrac{1}{\gamma}\!\right),\label{C_Z_1_Wait1}
\end{align}
where equality $(a)$ follows from $Y_{i}$ being independent of $Y_{i-1}$, $Z_i$, and $X_{i}$. Substituting \eqref{I_Z_1_Wait1}, \eqref{I-2_Z_1_Wait1}, and \eqref{C_Z_1_Wait1} into \eqref{A_AoI_Main}, we have Theorem~\ref{Age_Wait-1}.  
\begin{theorem}\label{Age_Wait-1}
For waiting functions
${Z_i
=(\beta-Y_{i-1})^+}$, the average AoI under the $\Wone$ policy is given as 
\begin{align}
&\Delta^{\mathrm{\Wone}} = \nn
&\quad\frac{\beta\mu (-\beta\gamma + \gamma - 2) - 2(\beta\gamma + 1)}{2(\gamma(\bar\mu^\beta + \beta\mu) + \mu)} + \beta + \frac{1}{\gamma} + \frac{2}{\mu} - 1.\label{eq_Wait1_beta}
\end{align}
\end{theorem}
Next, we provide an upper bound on the optimal $\beta$.
\begin{lemma}\label{lemma_beta_bound}
For a given service rate pair $(\gamma,\mu)$, $\beta^*$, the optimal value of $\beta$ for the $\Wone$ policy, satisfies $\beta^*\le\beta^{\max}$, where 
\begin{IEEEeqnarray}{rCl}
\beta^{\text{max}}&=&\bigg\lfloor{\dfrac{2\gamma+\sqrt{(\mu^2+\gamma\mu-2\gamma)^2+8(\mu^2+\gamma\mu)}}{2(\mu^2+\gamma\mu)}-\frac{1}{2}}\bigg\rfloor.\IEEEeqnarraynumspace\label{betamax2}
 \end{IEEEeqnarray}
\end{lemma}
\begin{proof}
    See Appendix \ref{appendixBBeta}.
\end{proof}

From Lemma~\eqref{lemma_beta_bound}, the optimal $\Wone$ policy is easily found by searching for the 
$\beta\in\{1,\ldots,\beta^{\text{max}}\}$ that minimizes \eqref{eq_Wait1_beta}. We also note that from the bound $\sqrt{x+y}\le\sqrt{x}+\sqrt{y}$ for non-negative $x,y$, it can be shown that 
\begin{align}
\beta^{\max} &\le \biggl\lfloor\Bigl[\frac{2\gamma}{\mu^2+\gamma\mu}-1
\Bigr]^{+} +\sqrt{\frac{2}{\mu^2+\gamma\mu}}
\biggr\rfloor\label{r_L_05}
\end{align}
From \eqref{r_L_05}, we see that $\beta^{\max}$ grows as $O(1/\mu)$ as $\mu\to0$.
By applying algebraic manipulations to the difference ${\Delta^{\mathrm{\Wone}}-\Delta^{\mathrm{\ZWone}}}$, we obtain the following corollary.
\begin{corollary}\label{Z_W_1&1_Packet}
In the $\text{1-Packet}$ system, it is optimal to wait, i.e., to set ${\beta\ge2}$ for the waiting function ${Z_i(Y_{i-1})=(\beta-Y_{i-1})^+}$, if the service rates $\mu$ and $\gamma$ satisfy
\begin{equation}
\mu\le\displaystyle\frac{-\gamma+\sqrt{\gamma^2+2\gamma}}{2} \quad\text{or}\quad\gamma\ge\displaystyle\frac{2\mu^2}{1-2\mu}.
\end{equation}
\end{corollary}
Corollary~\ref{Z_W_1&1_Packet} reveals a threshold: a $\Wone$ policy that reduces the average AoI {\em cannot} be realized if ${\mu>\frac{\sqrt{3}-1}{2}\approx0.366}$, regardless of the value of $\gamma$.

\section{Numerical Results}\label{Numerical Results}
In this section, we evaluate the performance of the optimal and fixed control policies proposed for the 1-Packet and 2-Packet systems as well as visualize the structure of the optimal policies. The optimal policies are obtained via the MDP-based approaches presented in Section \ref{sec_1P_control} and \ref{sec_2P_control}. The fixed policies are the $\ZWone$ policy (Section \ref{sec_1P_ZW}), $\ZWtwo$ policy (Section \ref{sec_2P_ZW}), and $\Wone$ policy (Section \ref{sec_1P_Wait}). Algorithm~\ref{RVI_algo} is run with the maximum value of the AoI $\bar\Delta=50$ and the stopping criterion threshold ${\epsilon=0.0005}$, unless otherwise stated.

\subsection{Structure of the Optimal Policies}
In this section, the structure of the optimal policy for the 1-Packet and 2-packet systems for service rate values $\mu\in\{0.2,0.3,0.4,0.5,0.6,0.7,0.8,0.9,1\}$ and $\gamma\in\{0.4,0.7,1\}$ are illustrated.

\subsubsection{1-Packet System}
Fig.~\ref{1-Packet_Structure} illustrates the structure of the optimal policy in the 1-Packet system by showing the optimal action in state ${s=(\delta,0,0,\star)}$ (i.e., empty system) for the different values of $\mu$ and $\gamma$. The figure reveals a \textit{threshold} structure with respect to the AoI, $\delta$: 
when the current AoI has a low value (e.g., when $\delta\le2$ for ${\gamma=0.4}$ and $\mu=0.2$), the optimal action is to stay idle (${a=0}$). Similarly, when the AoI exceeds this threshold, the optimal action is to send a request packet. When the reverse link service rate $\gamma$ increases, the AoI threshold increases. These behaviors are due to the fact that sending a new request packet would put the same burden on the network, regardless of its contribution to reducing the AoI, as any other request packet. Thus, if sending a new request packet cannot significantly reduce the AoI, it is optimal to wait. Finally, we observe that Corollary~\ref{Z_W_1&1_Packet} holds in Fig.~\ref{1-Packet_Structure}: for the sampler server's service rate $\mu$ exceeding $\frac{\sqrt{3}-1}{2}\approx0.366$, the optimal action is always ${a=1}$.

\begin{figure*}
\centering
\subfigure[$\gamma=0.4$ ]
{
\includegraphics[width=0.33\textwidth]{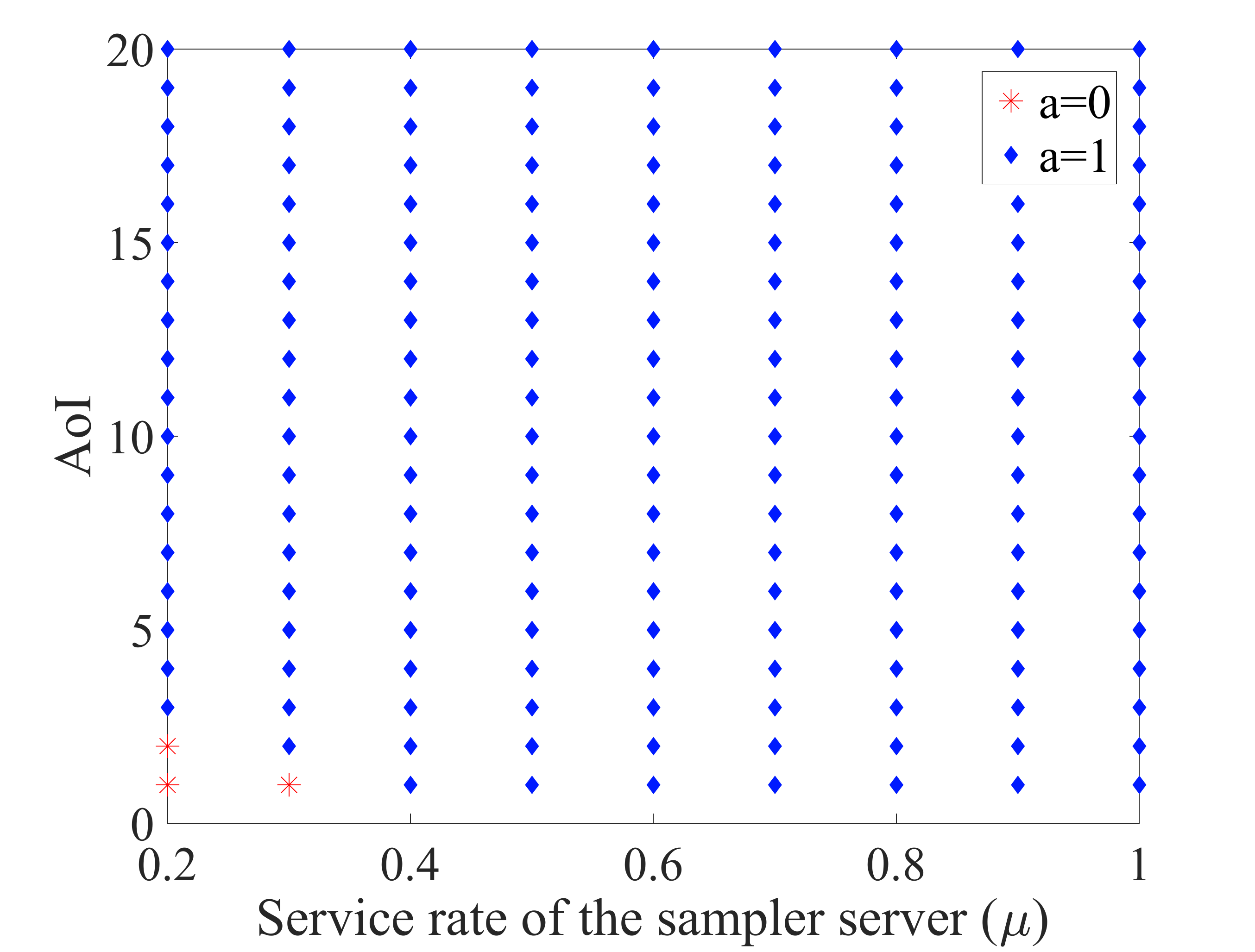}
\label{1_gamma_p4}\hspace{-6mm}
}
\subfigure[$\gamma=0.7$]
{
\includegraphics[width=0.33\textwidth]{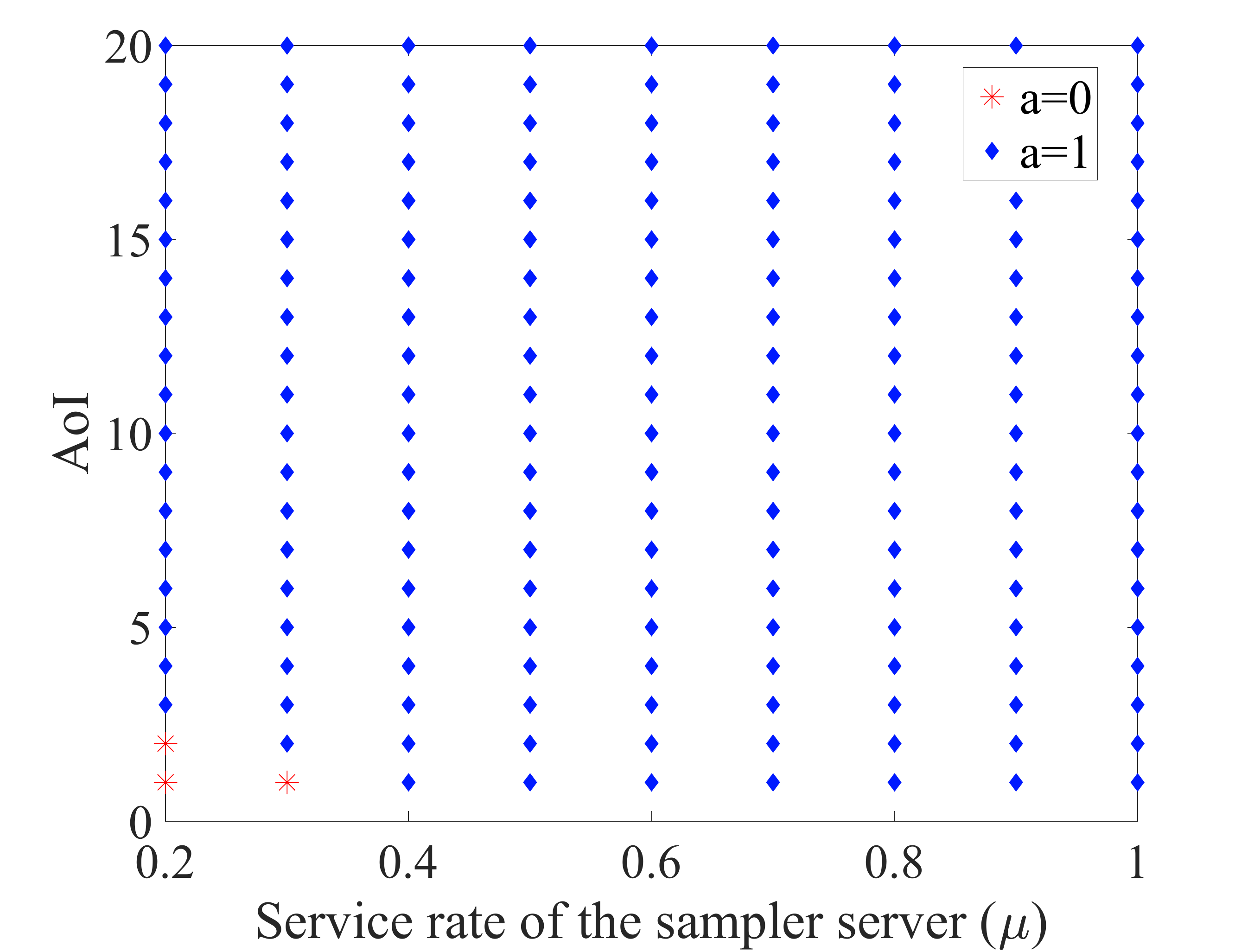}
\label{1_gamma_p7}\hspace{-6mm}
}
\subfigure[$\gamma=1$]
{
\includegraphics[width=0.32\textwidth]{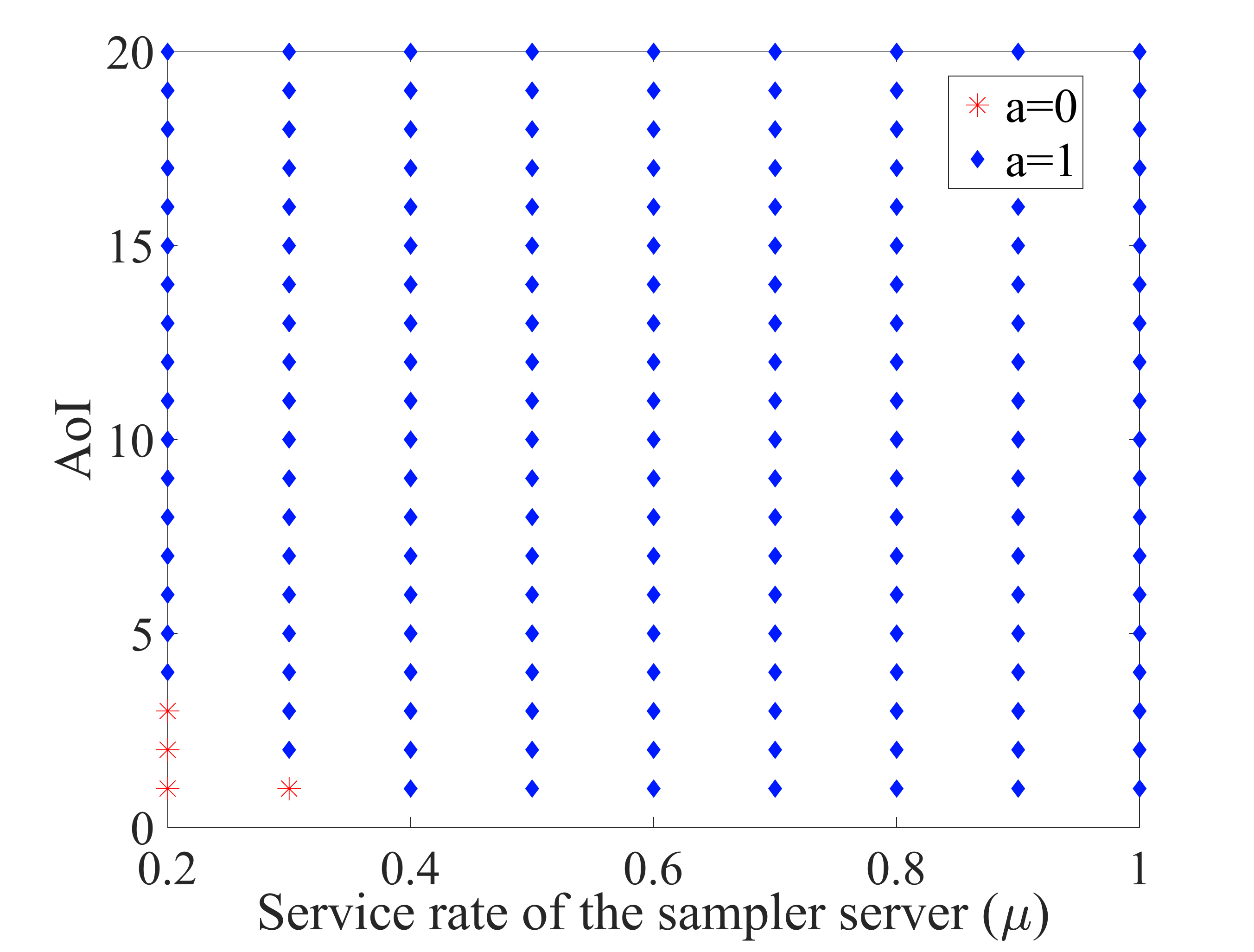}
\label{1_gamma_p10}
}
\vspace{-2mm}
\caption{Structure of the optimal policy for the 1-Packet system in state ${s=(\delta,0,0,\star)}$ and for the 2-Packet system in state ${\underline s=(\delta,0,0,0,0,\star,\star)}$; ${a=1}$ indicates that the optimal action is to generate a request packet, whereas ${a=0}$ indicates that the optimal action is to stay idle.}
\label{1-Packet_Structure}
\vspace{-5mm}
\end{figure*}

To further analyze the threshold structure, Fig.~\ref{Beta_p} illustrates the average AoI obtained by the $\Wone$ policy (Theorem~\ref{Age_Wait-1}) as a function of $\beta$ for different values of $\gamma$ with ${\mu=0.1}$. The average AoI obtained by the optimal policy in the 1-Packet system (obtained with RVI parameters ${\bar\Delta=100}$ and  $\epsilon=10^{-6}$) is also depicted. As it can be seen, when the value of $\gamma$ increases, the optimal value of $\beta$ increases, i.e., the optimal policy induces longer waiting times. This is because when the request packets are served faster in the reverse link, the corresponding generated status update packets will provide small AoI reduction if the AoI upon sending the new request was already small. Thus, it is beneficial to purposely postpone sending the requests. Note that the same trend is visible in Fig.~\ref{1-Packet_Structure}. Importantly, Fig.~\ref{Beta_p} shows that the average AoI obtained by the optimal $\Wone$ policy (i.e., the $\Wone$ policy with the optimal $\beta$) and the optimal policy in the 1-Packet system coincide. This shows evidence on the optimality of the waiting-based scheme, as conjectured in Section~\ref{sec_1P_Wait}.

\begin{figure*}
\centering
\subfigure[$\gamma=0.4$ ]
{
\includegraphics[width=0.33\textwidth]{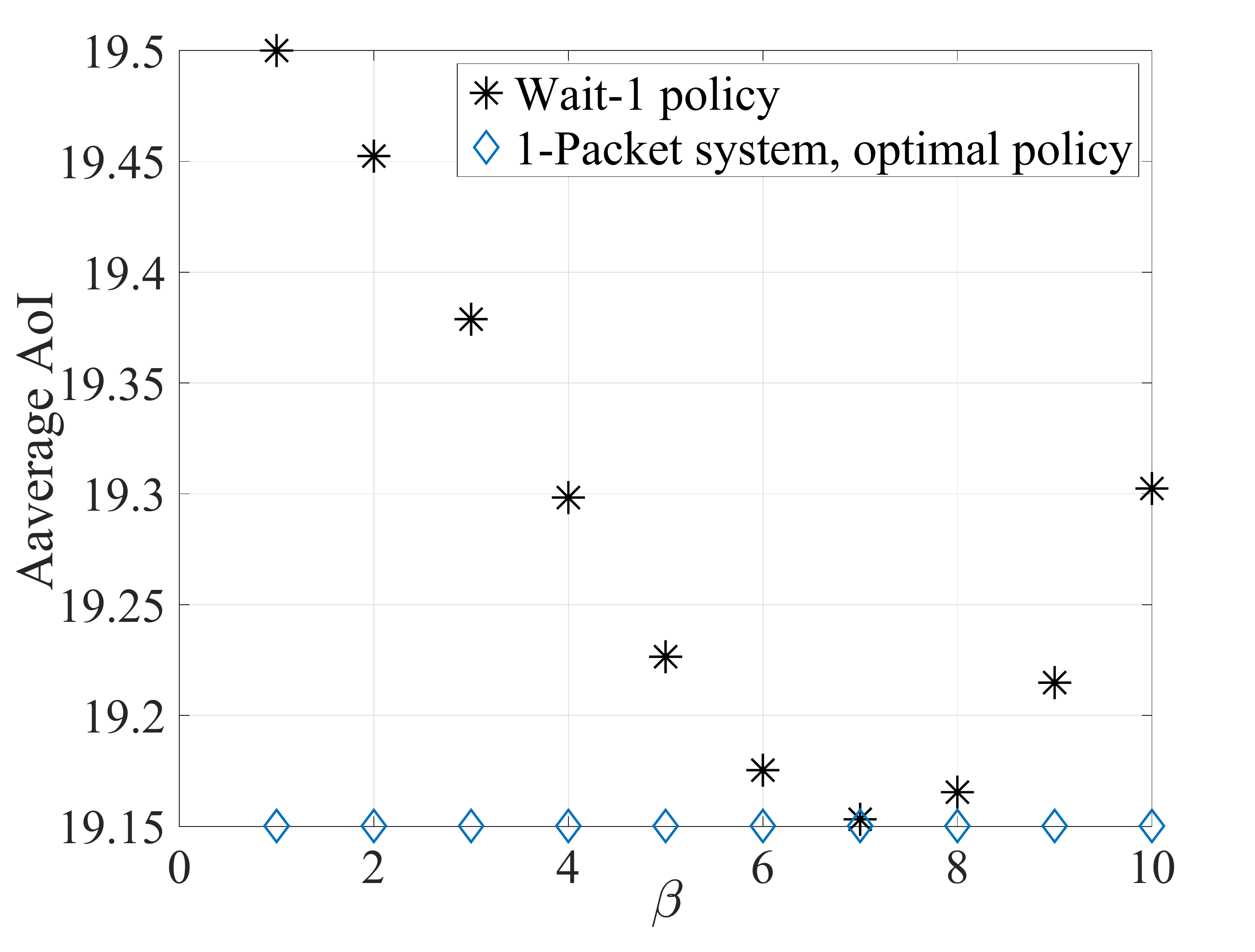}
\label{Beta_p4}\hspace{-6mm}
}
\subfigure[$\gamma=0.7$]
{
\includegraphics[width=0.33\textwidth]{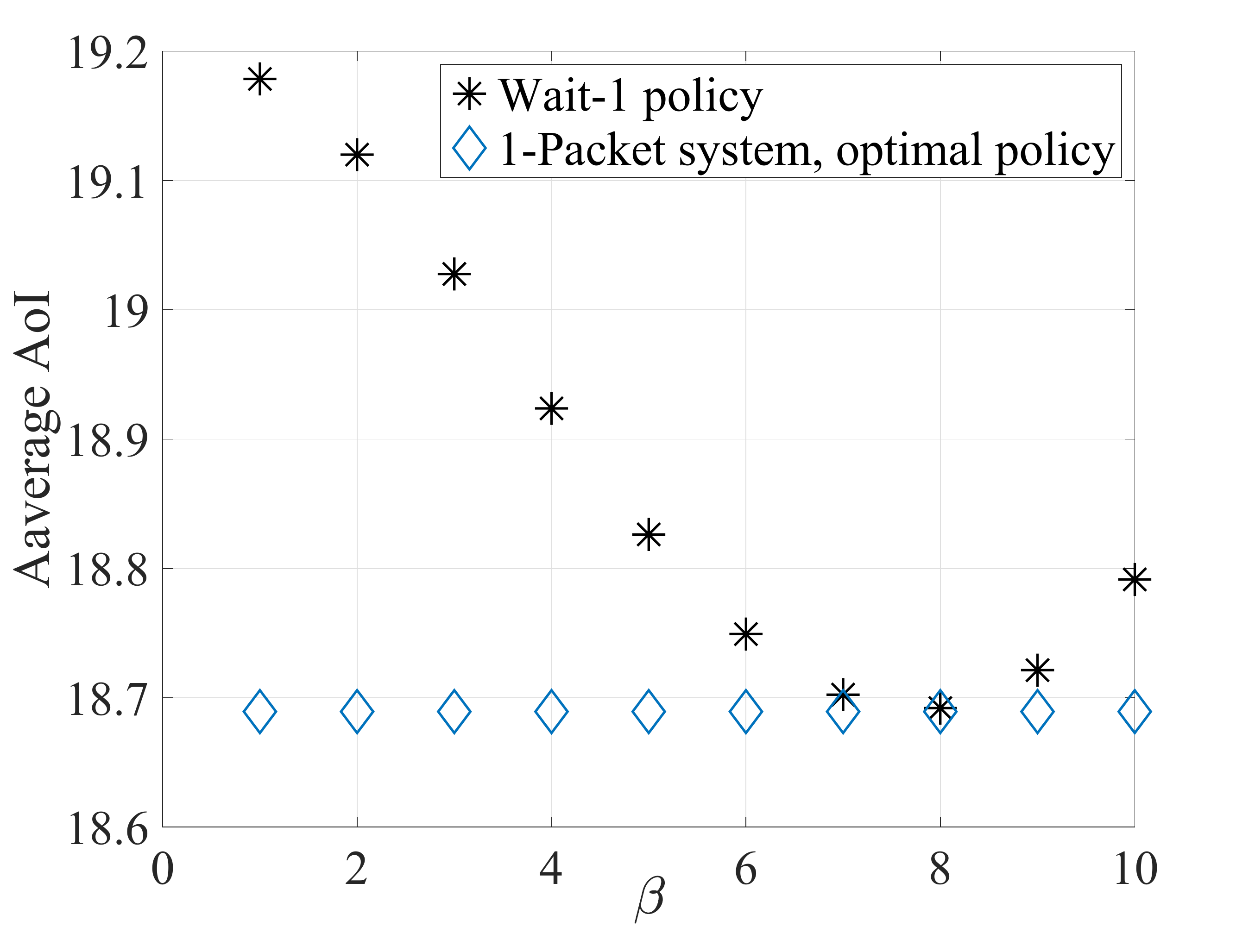}
\label{Beta_p7}\hspace{-6mm}
}
\subfigure[$\gamma=1$]
{
\includegraphics[width=0.33\textwidth]{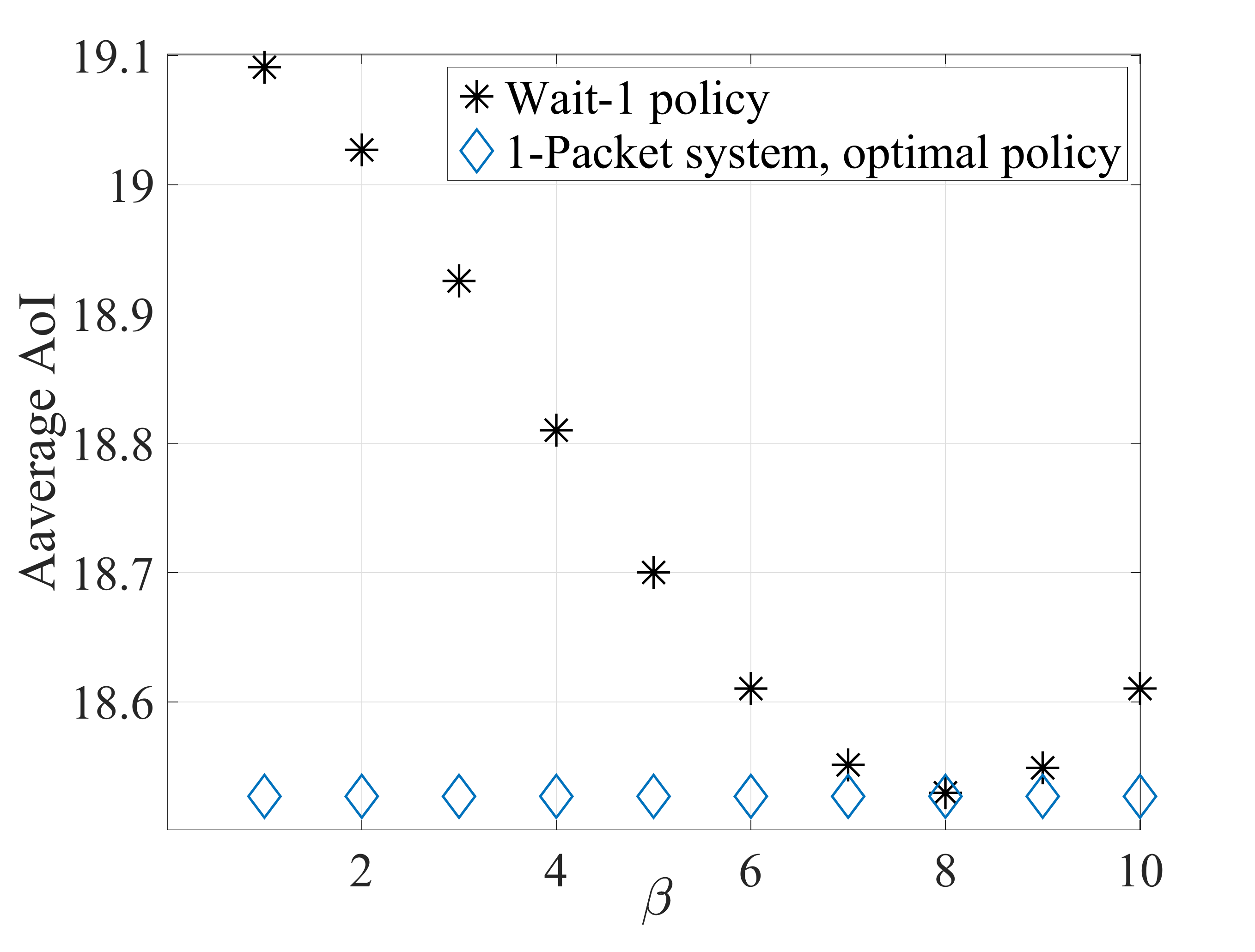}
\label{Beta_p10}
}
\vspace{-2mm}
\caption{Average AoI under the $\Wone$ policy as a function of $\beta$ for different values of $\gamma$ with $\mu=0.1$.
}
\label{Beta_p}
\end{figure*}

\subsubsection{2-Packet System}
In the 2-Packet system, the three states of interest where the optimal action needs to be determined are: 1) ${\underline s=(\delta,0,0,0,0,\star,\star)}$ (i.e., empty system), 2) ${\underline s=(\delta,0,0,0,1,\star,\Delta_{\mathrm{s}})}$ (i.e., there is a status update packet at the sampler server), and 3) ${\underline s=(\delta,0,1,0,0,\star,\Delta_{\mathrm{s}})}$ (i.e., there is a request packet at the controller server). Regarding the first case, the structure of the optimal policy in the 2-Packet system in state $\underline s=(\delta,0,0,0,0,\star,\star)$ is the same as that for the 1-Packet system; this threshold-type policy with respect to the AoI was illustrated in Fig.~\ref{1-Packet_Structure}. 
Regarding the third case, the optimal action in state $\underline s=(\delta,0,1,0,0,\star,\star)$ is to stay idle ($a=0$), regardless of the current value of the AoI. This stems from the perfect knowledge on the system's occupancy; see Remark~\ref{remark_state_aware}.  As the controller will immediately know when the reverse link server becomes empty, it can send a new request message just in time, should this action be optimal.
The second case is elaborated next.

Interestingly, we observe in Fig.~\ref{2s2-Packet_Structure} that in states ${\underline s=(\delta,0,0,0,1,\star,\Delta_{\mathrm{s}})}$, the optimal policy has a threshold with respect to the age $\Delta_{\mathrm{s}}$ of the packet at the sampler server. This behavior can be interpreted as follows. First, due to the memoryless property of the geometric random variable, the average residual service time of the packet in the sampler server is ${\mathbb{E}[Y]=1/\mu}$, regardless of $\Delta_{\mathrm{s}}$. Consequently, a newly inserted request message will lead its associated status update packet to enter the sampler queue with the same probability, regardless of $\Delta_{\mathrm{s}}$. Thus, when in state ${\underline s=(\delta,0,0,0,1,\star,\Delta_{\mathrm{s}})}$, the controller cannot take measures to avoid queueing delays in the sampler server, which always are detrimental to the age performance. However, the controller can look at the value of $\Delta_{\mathrm{s}}$ to optimize its action. If the status update packet under service has a low value of $\Delta_{\mathrm{s}}$, it will provide large age reduction once completing service. If a new request was sent in such state, it would, in turn, provide relatively low age reduction, because the preceding packet set the AoI to a low value. In this case, it is better to wait prior to sending a request, and hence the threshold in Fig.~\ref{2s2-Packet_Structure}. Furthermore, based on the same grounds as for Fig.\ \ref{1-Packet_Structure}, when the value of $\gamma$ increases, the optimal policy induces longer waiting times before sending a new request in state ${\underline s=(\delta,0,0,0,1,\star,\Delta_{\mathrm{s}})}$.

\begin{figure*}
\centering
\subfigure[$\gamma=0.4$ ]
{
\includegraphics[width=0.33\textwidth]{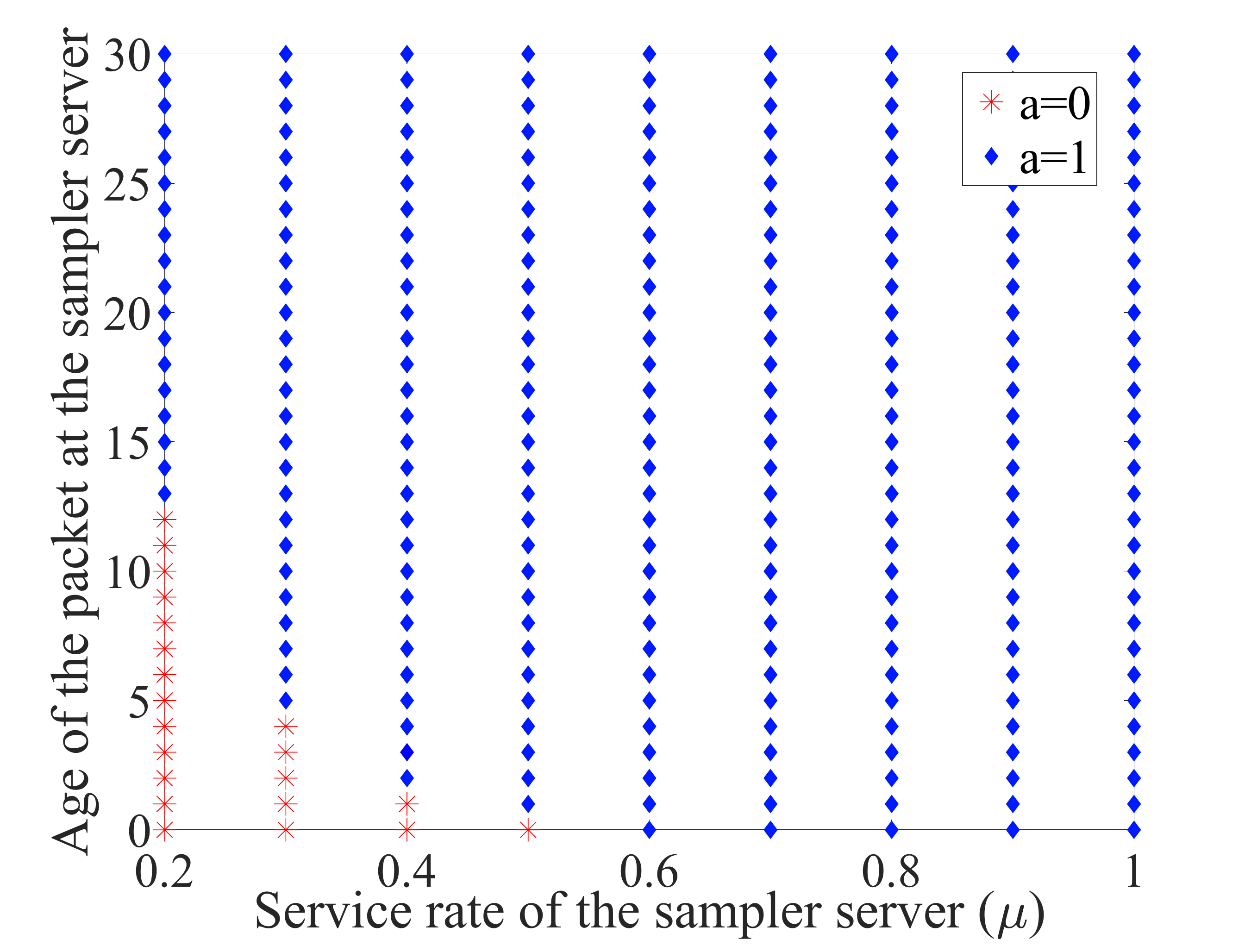}
\label{2s2_gamma_p4}\hspace{-6mm}
}
 \subfigure[$\gamma=0.7$]
 {
 \includegraphics[width=0.33\textwidth]{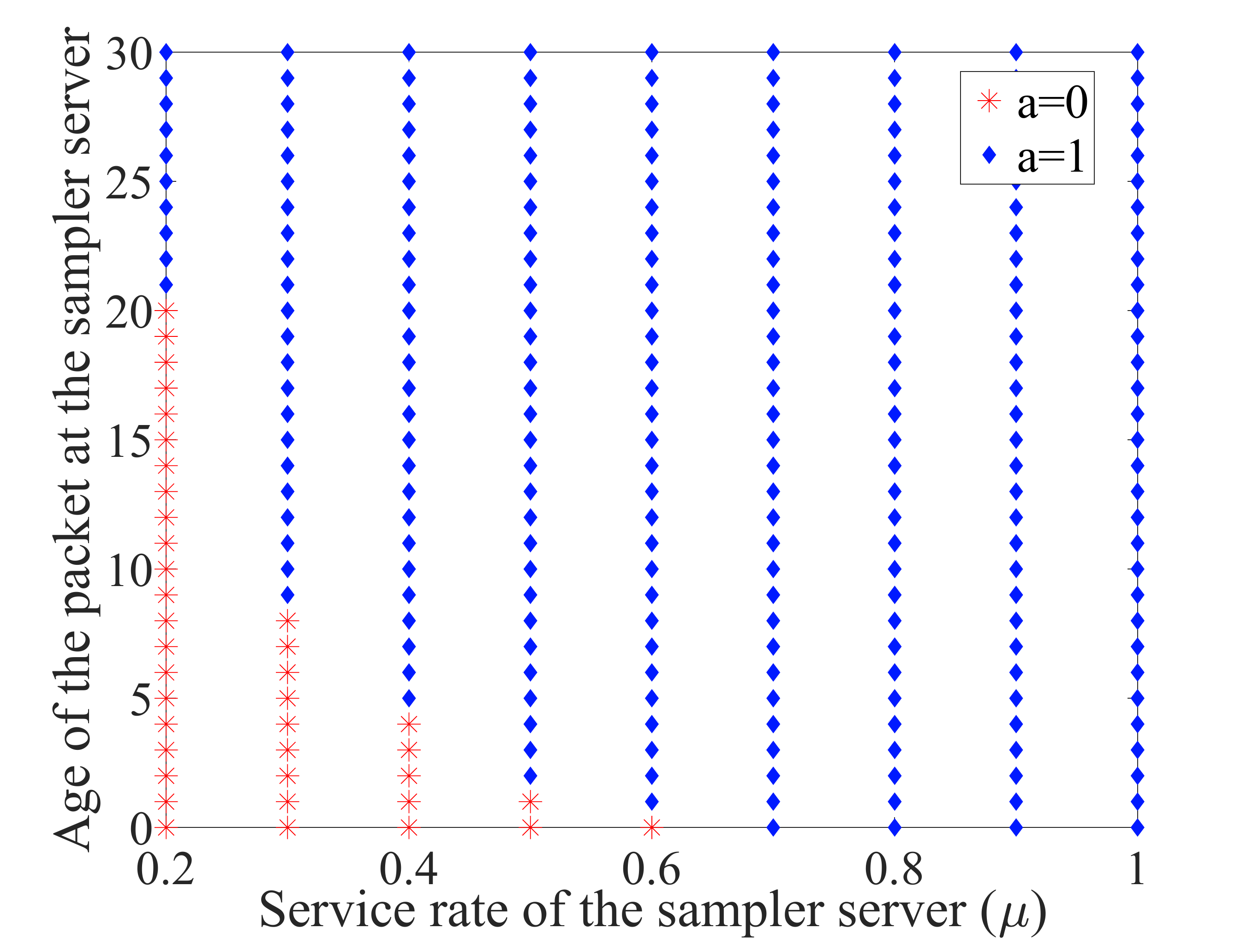}
 \label{2s2_gamma_p7}\hspace{-6mm}
 }
 \subfigure[$\gamma=1$]
 {
 \includegraphics[width=0.33\textwidth]{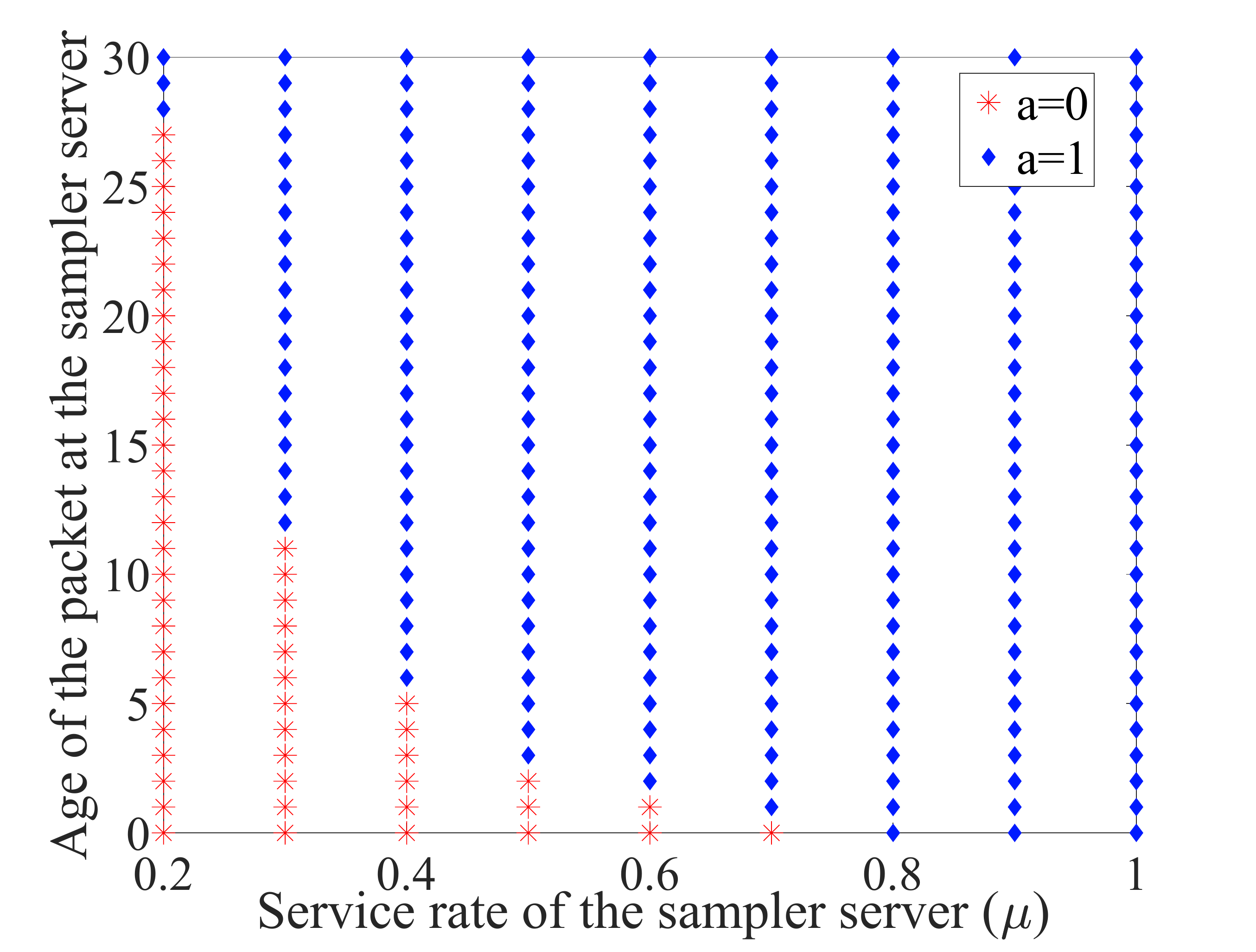}
 \label{2s2_gamma_p10}
 }
\vspace{-2mm}
\caption{Structure of the optimal policy for the 2-Packet system in state $\underline s=(\delta,0,0,0,1,\star,\Delta_{\mathrm{s}})$.}
\label{2s2-Packet_Structure}
\vspace{-5mm}
\end{figure*}

In practice, we select the AoI upper-bound $\bar{\Delta}$ by increasing its value until the average AoI of an optimal policy does not change. To illustrate this behavior, Fig.~\ref{AgebarD} shows the average AoI of an optimal policy for the 1-Packet system as a function of $\bar{\Delta}$ for different service rate values $\mu$ and $\gamma$. As it can be seen, if we set $\bar{\Delta}$ to a too small value, the average AoI value will substantially change -- more precisely, decrease. In consequence, one would draw wrong conclusions about the system's performance. To give some numbers, the figure shows that for the service rates $\mu=0.2,~\gamma=0.4$, an age upper-bound value $\bar{\Delta}>45$ provides sufficient accuracy for practical purposes. Similarly, for $\mu=0.4,~\gamma=0.4$, this saturation point occurs for age upper-bound values $\bar{\Delta}>25$.

\begin{figure}[t]
\centering
\includegraphics[width=0.4\textwidth]{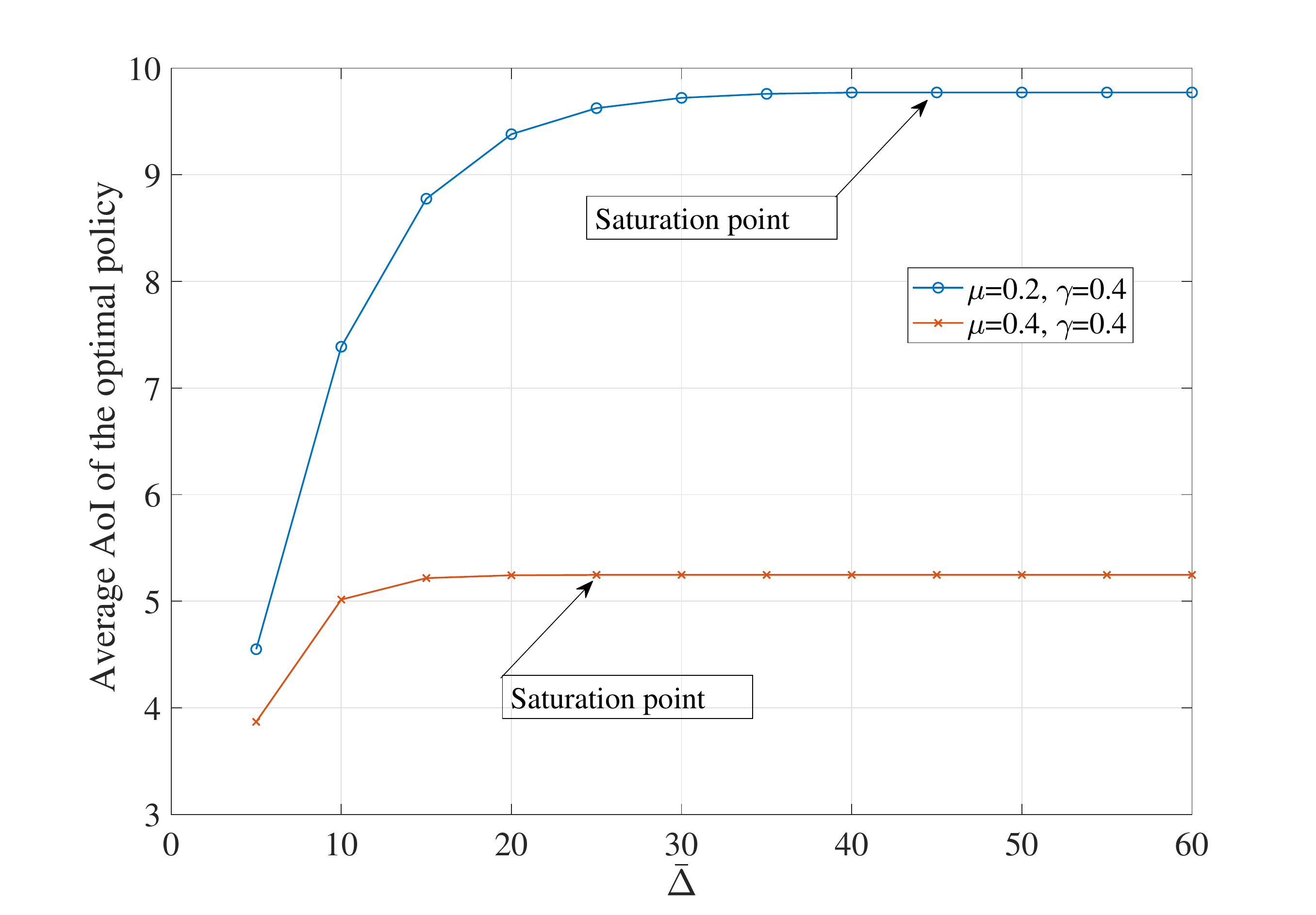}\vspace{-2mm}
\caption{Average AoI of the optimal policy for the 1-Packet system as a function of $\bar\Delta$ under different values of ${\mu}$ and ${\gamma}$.}
\label{AgebarD}
\vspace{-5mm}
\end{figure}

\subsection{Average AoI Value for Different Policies}
Fig.~\ref{Average_AoI_val} illustrates the value of the average AoI as a function of $\mu$ for different values of $\gamma$ under the $\ZWone$ policy, the $\ZWtwo$ policy, and the optimal policies in the 1-Packet and 2-packet systems. As it can be seen, the optimal policy in the 2-Packet system, in general, outperforms the other policies, or, for certain service rates, coincides with another policy. This is as expected, because the policy optimization is over a set of policies for the 2-Packet system includes these other policies.

Fig.~\ref{Average_AoI_val} shows that for 
${\mu\ge\sqrt{2}/2\approx0.7071}$, 
$\ZWtwo$ 
outperforms 
$\ZWone$
regardless of the value of $\gamma$, as shown in Corollary~\ref{Z_W_2&Z_W_1}. Moreover, we can see that for 
${\mu\ge\frac{\sqrt{3}-1}{2}\approx0.366}$, the $\ZWone$ policy and the optimal policy in the 1-Packet system provide the same 
average AoI regardless of the value of $\gamma$, as shown in Corollary~\ref{Z_W_1&1_Packet}.

We can also see from Fig.~\ref{Average_AoI_val} that as $\mu$ increases, the average AoI decreases for all policies, as expected, because the status update packets are served faster. However, maybe a bit surprisingly, such monotonic behavior is not seen with respect to $\gamma$; increasing the value of $\gamma$ is not always beneficial.
As it can be seen, for the low values of $\mu$, the $\ZWtwo$ policy results in a higher average AoI value when the value of $\gamma$ increases. This occurs because 
a large value of $\gamma$ causes the requests to be served too fast in relation to the sampler service rate $\mu$, leading to queueing of update packets in the sampler buffer, deteriorating the AoI performance.

\begin{figure}[t]
\centering
\subfigure[$\gamma=0.4$ ]
{
\includegraphics[width=0.4\textwidth]{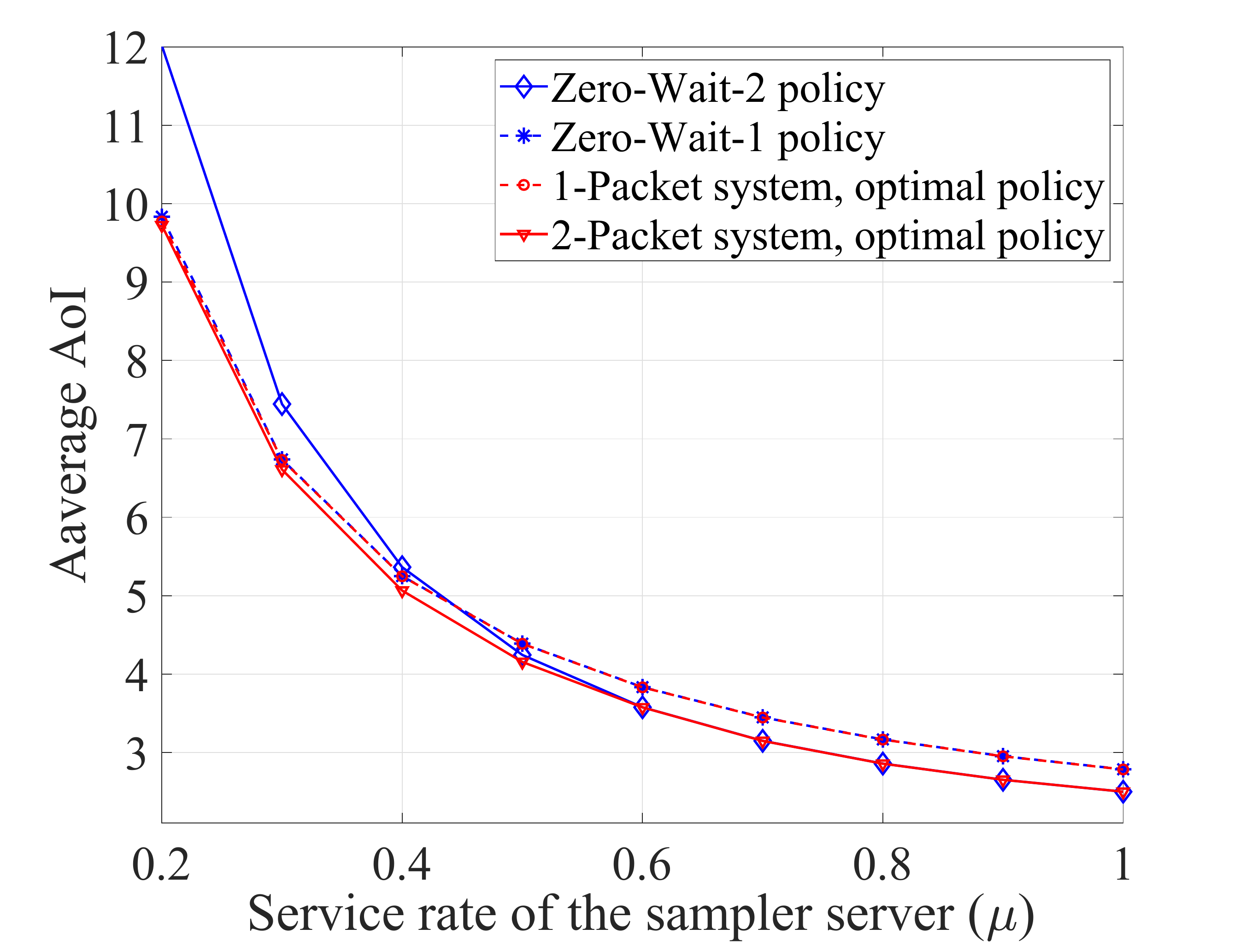}
\label{Average_AoI_val_p4}\vspace{-2mm}
}
 \subfigure[$\gamma=0.7$]
 {
 \includegraphics[width=0.4\textwidth]{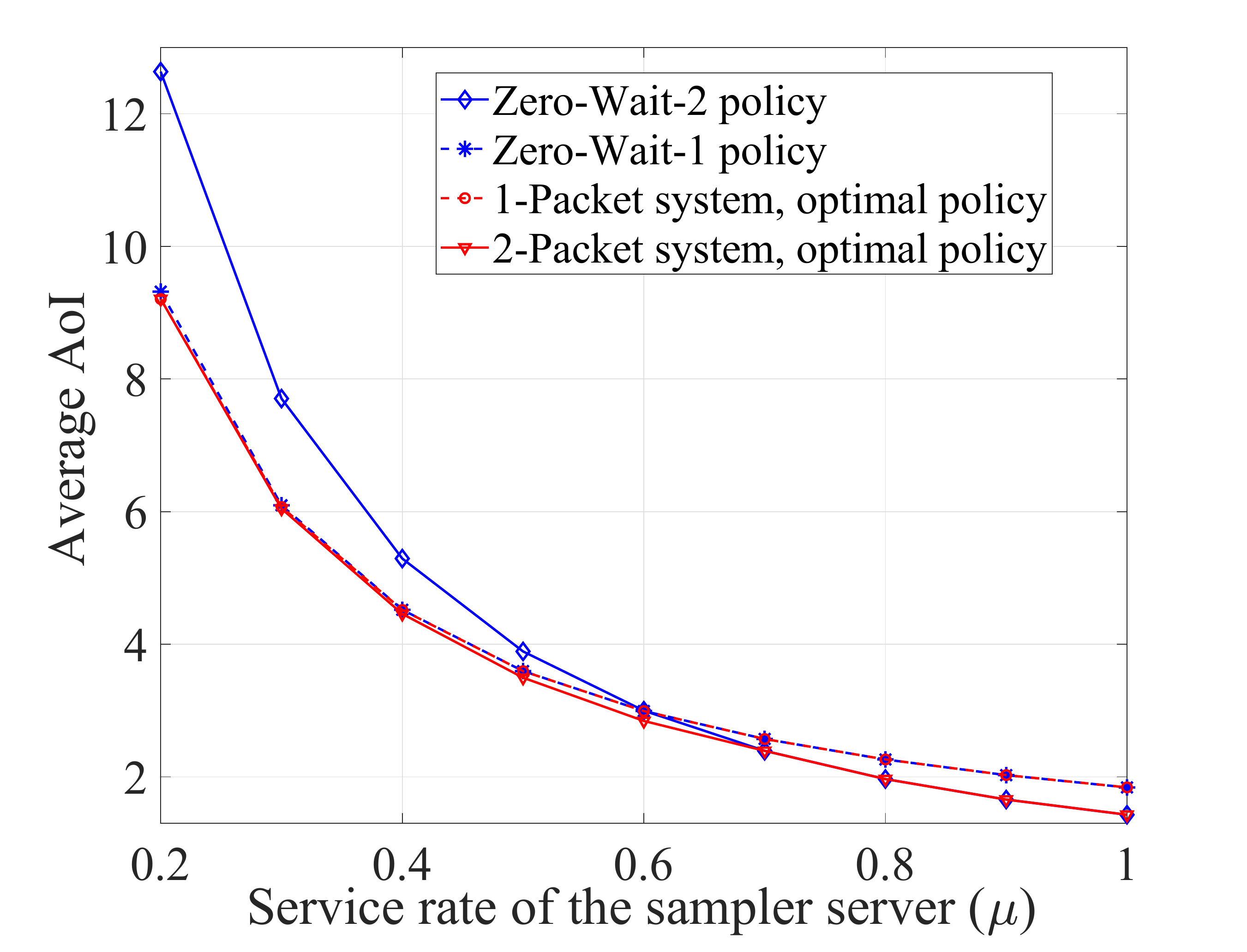}
 \label{Average_AoI_val_p7}\vspace{-2mm}
 }
 \subfigure[$\gamma=1$]
 {
 \includegraphics[width=0.4\textwidth]{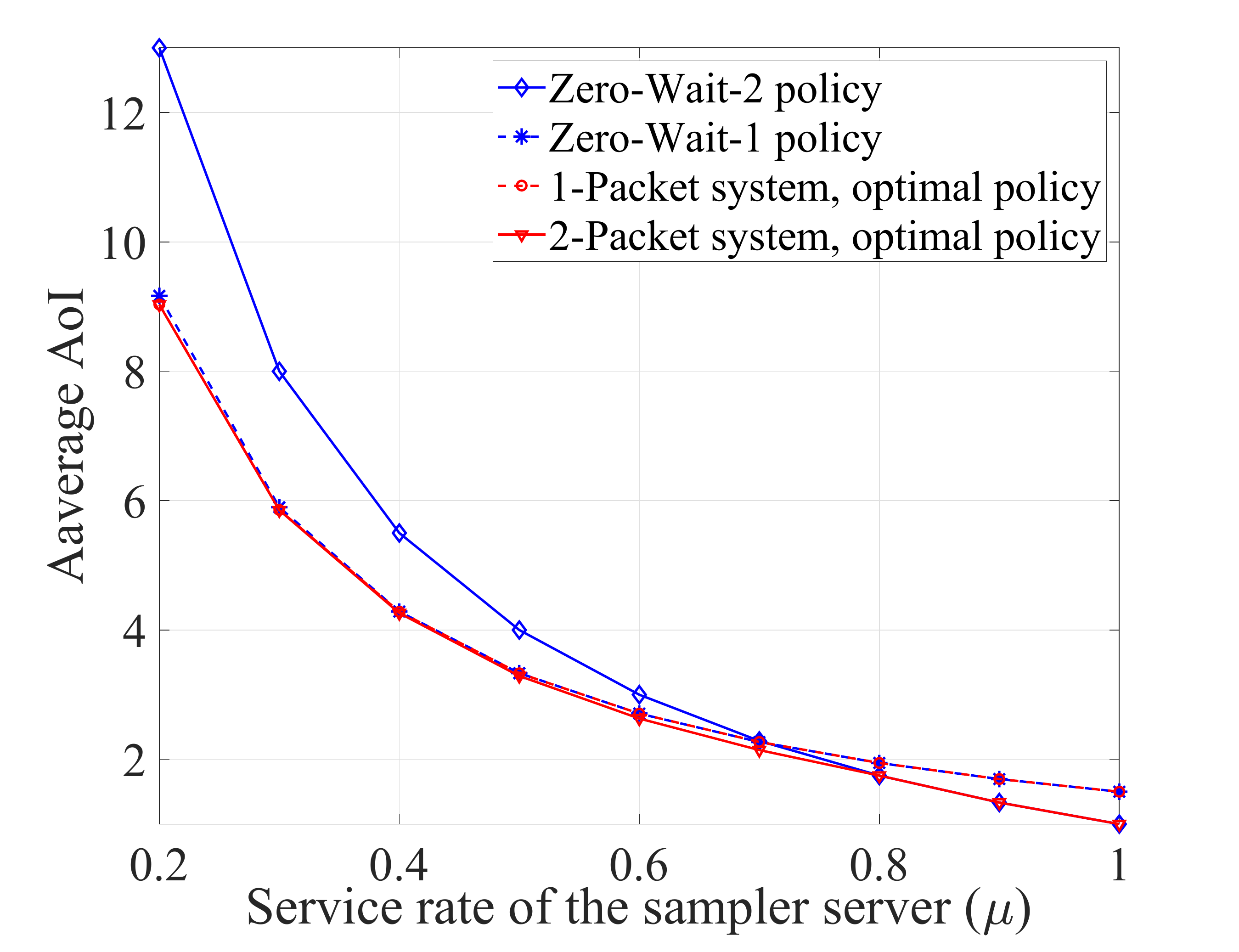}
 \label{Average_AoI_val_p10}
 }
\vspace{-2mm}
\caption{Average AoI as a function of $\mu$ for different values of $\gamma$.}
\label{Average_AoI_val}
\vspace{-5mm}
\end{figure}

\section{Conclusions} \label{Conclusions}
We studied status updating under two-way delay in a discrete-time system. 
We developed AoI-optimal control policies using the tools of MDPs in two scenarios. We began with the system having at most one active request. Then, we initiated pipelining-type status updating by studying a system having at most two active requests.
In addition, we conducted AoI analysis by deriving closed-form expressions of the average AoI under the $\ZWone$,   $\ZWtwo$ and $\Wone$ policies which employ only the controller's local information. 

Numerical results showed that 
for both 1-Packet and 2-Packet systems, 
the optimal policy is a threshold-type policy. In addition, it was seen that for high values of the update service rate $\mu$,
$\ZWone$ is an optimal policy for the  1-Packet system, and $\ZWtwo$ is an optimal policy for the 2-Packet system, regardless of the service rate of the controller server. Moreover, the results showed that increasing the request message service rate $\gamma$ is not necessarily beneficial for all policies, e.g., when the update service rate $\mu$ is low, the average AoI under the $\ZWtwo$ policy is increasing 
$\gamma$.

The interesting future work would include the study of i) optimal policies to minimize non-linear functions of the AoI for arbitrary service time distributions for controller and sampler servers, ii) the structure of optimal policies in the 1-Packet and 2-Packet systems (our numerical results showed that the optimal policies appear to have a threshold-based structure, which, however, remains to be proven),
and iii) the performance of the two-way delay system using real-world data sets.

\begin{figure*}[t]
\centering
\includegraphics[width=.65\linewidth,trim = 0mm 3mm 0mm 0mm,clip]{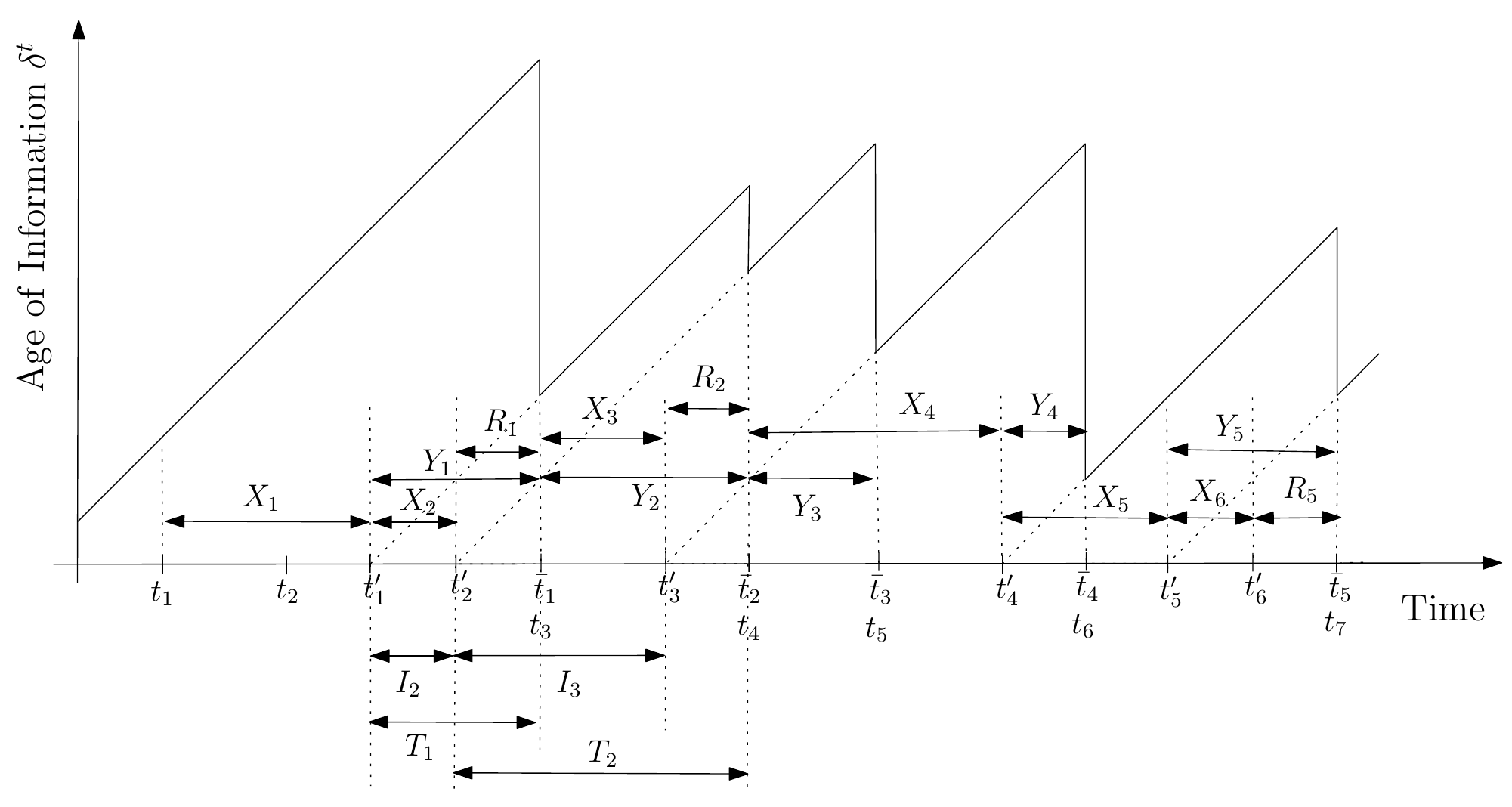}\vspace{-2mm}
\caption{AoI as a function of time under the $\ZWtwo$ policy.}\vspace{-3mm}
\label{fig:Zero-W-2}
\end{figure*}


\section{Appendix}
\subsection{Proof of Theorem~\ref{thm:ZW2}}\label{ZW2-proof}
Referring to the \ZW{2} AoI sample path depicted in Figure~\ref{fig:Zero-W-2}, we start with the arrival of update packet $i-1$ at the sampler server and we use the partition $\set{B_{i-1},\Bbar_{i-1}}$  to describe what can follow:  
\begin{itemize}
\item If $\Bbar_{i-1}$, then update $i-1$ immediately goes into service at the sampler server, starting service time $Y_{i-1}$. Request message $i$ simultaneously starts service time $X_i$ at the controller server.
\begin{itemize}
\item If $X_i\ge Y_{i-1}$, then event $\Bbar_i$ occurs since update $i$ arrives at the sampler server after update $i-1$ departs. Thus, update packet $i$ arrives at an idle sampler server and immediately starts service time $Y_i$. In this case, $I_i=X_i$ and $T_i=Y_i$.
\item If $B_i=\set{X_i<Y_{i-1}}$ occurs,
then update $i$ waits 
for the residual service time $R_{i-1}$ of update $i-1$ before starting its own service time $Y_i$. In this case, $I_i=X_i$ and $T_i=R_{i-1}+Y_i$. 
\end{itemize}
\item If $B_{i-1}$, update $i-1$ arrives when the sampler server is busy and waits for the residual service time $R_{i-2}$ of update $i-2$. After this waiting time, update $i-2$ is delivered to the sink, request $i$ is generated at the controller and starts service time $X_i$, {\em and} update $i-1$ simultaneously starts service time $Y_{i-1}$ at the sampler server. 
\begin{itemize}
\item If $\Bbar_i=\set{
X_i\ge Y_{i-1}}$ occurs, update packet $i$ arrives at an empty sampler server and immediately starts service time $Y_i$. In this case, 
$I_i=R_{i-2}+X_i$ and $T_i=Y_i$.
\item If $B_i=\set{X_i<Y_{i-1}}$ occurs, update $i$ arrives at the sampler server before update $i-1$ departs. In this case, update $i$ waits at the sampler server for the residual service time $R_{i-1}$ of update $i-1$, implying 
$I_i=R_{i-2}+X_i$ and $T_i=R_{i-1}+Y_i$. 
\end{itemize}
\end{itemize}
To summarize this two-step partitioning, we see that
\begin{subequations}\label{BBpartition}
\begin{align}
\text{if $\Bbar_{i-1}\Bbar_i$:}\quad I_i&=X_i, &  T_i&=Y_i;\\
\text{if $\Bbar_{i-1}B_i$:}\quad I_i&= X_i, & T_i &=R_{i-1}+Y_i;\\
\text{if $B_{i-1}\Bbar_i$:}\quad I_i&= R_{i-2}+X_i, & T_i&=Y_i;\\
\text{if $B_{i-1}B_i$:}\quad I_i&=R_{i-2}+X_i, & T_i &=R_{i-1}+Y_i.
\end{align}
\end{subequations}
Note that the memoryless property of geometric random variables implies that conditioned on $B_j$, each $R_{j-1}$ is a geometric $(\mu)$ random variable, independent of all $X_i$ and $\set{Y_i\colon i\neq j-1}$. 

Furthermore, careful reading of this two-step partitioning shows that $B_i$ and $B_{i-1}$ are  independent events since ${B_i=\set{X_i<Y_{i-1}}}$, whether  ${B_{i-1}=\set{X_{i-1}<Y_{i-2}}}$ occurred or not in the prior cycle. Since the $X_i$ and $Y_i$ sequences are independent (and iid within each sequence), 
\begin{align}
\prob{B_i}=\prob{B_i|B_{i-1}}
=\prob{X_i<Y_{i-1}}.
\end{align}
We further observe that the occurrence of $B_i$ suggests that $X_i$ is shorter than usual. 
In Appendix~\ref{appendix-ZW2-XB-props}, we verify the following properties.

\begin{lemma}\label{ZW2-XB-props}
\begin{subequations}
In the $\ZW{2}$ system: 
\begin{enumerate}
\item[(a)]
Update $i$ arrives at a busy sampler server with probability 
\begin{equation}\label{new-pro_X_Y}
\prob{B_i}=\pb\triangleq\Pr(X_i < Y_{i-1})=\displaystyle\frac{\gamma\bar\mu}{\mu+\gamma\bar{\mu}}.
\end{equation}
\item[(b)] Given event  $B_i$, 
\begin{equation}
    \E{X_i|B_i} =\frac{1}{\mu+\gamma\bar\mu}=\frac{\pb}{\gamma\bar\mu}.
\end{equation}
\end{enumerate}
\end{subequations}
\end{lemma}

Since \eqref{BBpartition} shows that $I_i$ depends only on the partition $\set{B_{i-1},\Bbar_{i-1}}$, it follows from the law of total expectation that
\begin{align}
&\E{I_i}=\E{I_i|\Bbar_{i-1}}\prob{\Bbar_{i-1}}
+\E{I_i|B_{i-1}}\prob{B_{i-1}}\nn
&=\E{X_i|\Bbar_{i-1}}\prob{\Bbar_{i-1}}
+\E{R_{i-2}+X_i|B_{i-1}}\prob{B_{i-1}}\nn
&=\E{X_i}+\E{R_{i-2}|B_{i-1}}\prob{B_{i-1}}\nn
&=\frac{1}{\gamma} +\frac{\pb}{\mu}.
\label{eqn:ZW2-EI}
\end{align}
Similarly for the second moment,
\begin{IEEEeqnarray}{rCl}
\E{I_i^2}&=&\E{I^2_i|\Bbar_{i-1}}\prob{\Bbar_{i-1}}
+\E{I^2_i|B_{i-1}}\prob{B_{i-1}}\nn
&=&\E{X_i^2|\Bbar_{i-1}}\prob{\Bbar_{i-1}}\nn
&&\quad+\E{R_{i-2}^2+2X_iR_{i-2} +X_i^2|B_{i-1}}\prob{B_{i-1}}\nn
&=&\E{X^2_i}+\E{R_{i-2}^2+2X_iR_{i-2}|B_{i-1}}\prob{B_{i-1}}.\IEEEeqnarraynumspace
\label{eqn:ZW2-EI2a}
\end{IEEEeqnarray}
Since $X_i$ is geometric $(\gamma)$, independent of the event $B_{i-1}$, and  $R_{i-1}$ is geometric $(\mu)$ given $B_{i-1}$, 
it follows from \eqref{eqn:ZW2-EI} and \eqref{eqn:ZW2-EI2a} that
\begin{align}
\E{I_i^2} &=\frac{2-\gamma}{\gamma^2}+
\paren{\frac{2-\mu}{\mu^2} +\frac{2}{\gamma\mu}}\pb\nn
&=\paren{\frac{2}{\gamma}-1}\paren{\frac{1}{\gamma}+\frac{\pb}{\mu}}+2\frac{\pb}{\mu^2}\nn
&=\paren{\frac{2}{\gamma}-1}\E{I}+2\frac{\pb}{\mu^2}.\label{eqn:EI2}
\end{align}
We employ the partition $\set{\Bbar_{i-1}\Bbar_i,\Bbar_{i-1}B_i,B_{i-1}\Bbar_i,B_{i-1}B_i}$ and the law of total expectation to evaluate
\begin{align}
\E{I_iT_i}
&=\E{I_iT_i|\Bbar_{i-1}\Bbar_i}
\prob{\Bbar_{i-1}\Bbar_i}\nn
&\qquad+\E{I_iT_i|\Bbar_{i-1}B_i}\prob{\Bbar_{i-1}B_i}\nn
&\qquad+\E{I_iT_i|B_{i-1}\Bbar_i}\prob{B_{i-1}\Bbar_i}\nn
&\qquad
+\E{I_iT_i|B_{i-1}B_i}\prob{B_{i-1}B_i}.
\end{align}
It follows from \eqref{BBpartition} and the law of total expectation that
\begin{IEEEeqnarray}{lCl}
\mathrlap{\E{I_iT_i}}&& \nn
&=&\E{X_iY_i|\Bbar_{i-1}\Bbar_i}
\prob{\Bbar_{i-1}\Bbar_i}\nn
&&+\E{X_i(R_{i-1}+Y_i)|\Bbar_{i-1}B_i}
\prob{\Bbar_{i-1}B_i}\nn
&&+\E{(R_{i-2}+X_i)Y_i|B_{i-1}\Bbar_i}
\prob{B_{i-1}\Bbar_i}\nn
&&+
\E{(R_{i-2}+X_i)(R_{i-1}+Y_i)|B_{i-1}B_i}
\prob{B_{i-1}B_i}\\
&=&\E{X_iY_i} +\E{R_{i-1}X_i|\Bbar_{i-1}B_i}
\prob{\Bbar_{i-1}B_i}\nn
&&+\E{R_{i-2}Y_i|B_{i-1}\Bbar_i}
\prob{B_{i-1}\Bbar_i}\nn
&&+
\E{R_{i-2}(R_{i-1}\!+\!Y_i)\!+\!R_{i-1}X_i|B_{i-1}B_i}
\prob{B_{i-1}B_i}.\IEEEeqnarraynumspace
\label{eqn:EITv3}
\end{IEEEeqnarray}
Note that $Y_i$ is independent of $B_{i-1}$ and $B_i$ while  $X_i$ is independent of $B_{i-1}$ but  influenced by the occurrence of $B_i$. Since $\E{X_iY_i}=1/(\gamma\mu)$, it follows from \eqref{eqn:EITv3} that \begin{align}
\E{I_iT_i} &=\frac{1}{\gamma\mu} 
+\frac{\E{X_i|B_i}
\prob{\Bbar_{i-1}B_i}}{\mu}
+\frac{\prob{B_{i-1}\Bbar_i}}{\mu^2}\nn
&\qquad+\paren{\frac{1}{\mu^2}+\frac{1}{\mu^2} +\frac{\E{X_i|B_i}}{\mu}}\prob{B_{i-1}B_i}\\
&=\frac{1}{\gamma\mu}
+\frac{2\prob{B_{i-1}B_i}+\prob{B_{i-1}\Bbar_i}}{\mu^2}\nn
&\qquad+\frac{\E{X_i|B_i}
(\prob{\Bbar_{i-1}B_i}+\prob{B_{i-1}B_i})}{\mu}\\
&=\frac{1}{\gamma\mu}
+\frac{\prob{B_{i-1}B_i}+\prob{B_{i-1}}}{\mu^2}
+\frac{\E{X_i|B_i}\pb}{\mu}\nn
&=\frac{1}{\gamma\mu}
+\frac{\pb(1+\pb)}{\mu^2}
+\frac{\E{X_i|B_i}\pb}{\mu}.\label{eqn:EIT-v4}
\end{align}
Applying Lemma~\ref{ZW2-XB-props} to 
\eqref{eqn:EIT-v4} yields
\begin{align}
\E{I_iT_i} 
&=\frac{1}{\gamma\mu}
+\frac{\pb(1+\pb)}{\mu^2}
+\frac{\pb^2}{\gamma\mu\bar{\mu}}\nn
&=\frac{1}{\mu}\paren{\frac{1}{\gamma}+\frac{\pb}{\mu}}+\frac{\pb^2}{\mu}\paren{\frac{1}{\mu}+\frac{1}{\gamma\bar{\mu}}}\nn
&=\frac{\E{I_i}}{\mu}+\frac{\pb^2}{\mu}\paren{\frac{1}{\mu}+\frac{1}{\gamma\bar{\mu}}}\nn
&=\frac{\E{I_i}}{\mu}+\frac{\pb^2}{\mu}\paren{\frac{\mu+\gamma\bar{\mu}}{\gamma\mu\bar{\mu}}}
=\frac{\E{I_i}}{\mu}+\frac{\pb}{\mu^2}.
\label{eqn:EIT}
\end{align}

Combining \eqref{A_AoI_Main}, \eqref{eqn:EI2} and \eqref{eqn:EIT}, we obtain
\begin{align}
    \Delta^{\ZWtwo}
    &=\frac{1}{\gamma}+\frac{1}{\mu}-1 +2\frac{\pb}{\mu^2\E{I}}.
\end{align}
Applying  Lemma~\ref{ZW2-XB-props}(a) and     \eqref{eqn:ZW2-EI} yields Theorem~\ref{thm:ZW2}. 

\subsection{Proof of Lemma~\ref{ZW2-XB-props}}
\label{appendix-ZW2-XB-props}
\noindent (a) The probability of event $B_{i}=\set{X_{i}<Y_{i-1}}$ is
\begin{align}
\pb&=\sum_{j=1}^{\infty}\Pr(X_i < Y_{i-1} \mid X_i=j)\Pr(X_i=j)\nn
&\overset{(a)}{=}
\sum_{j=1}^{\infty}\bar\mu^j\gamma\bar\gamma^{j-1}=\frac{\gamma}{\bar\gamma}\sum_{j=1}^{\infty}\big(\bar\mu\bar\gamma\big)^{j}\overset{(b)}{=}\dfrac{\gamma\bar\mu}{\mu+\gamma\bar\mu}.
\end{align}
where $(a)$ holds because ${\prob{X_i=j}=\gamma\bar\gamma^{j-1}}$ and 
$\prob{j < Y_{i-1}}
=\bar\mu^j$, 
and step $(b)$ follows from the identities $\sum_{j=1}^{\infty}a^{j}={a}/{(1-a)}$ for  $\abs{a}<1$, and  $1-\bar\mu\bar\gamma=\mu+\gamma\bar\mu$.

\noindent (b) For the conditional expectation, similar facts imply
\begin{align}
\E{X_i|B_i} 
&=\sum_{j=1}^\infty j\prob{X_i=j|X_i<Y_{i-1}}\nn
&=\frac{1}{\pb}\sum_{j=1}^\infty j\prob{X_i=j,Y_{i-1}>j}\nn
&=\frac{1}{\pb}\sum_{j=1}^\infty j\bar{\gamma}^{j-1}\gamma\bar{\mu}^{j}\nn
&=\frac{\gamma\bar{\mu}}{\pb(1-\bar{\gamma}\bar{\mu})}\sum_{j=1}^\infty j(\bar{\gamma}\bar{\mu})^{j-1}(1-\bar{\gamma}\bar{\mu})\nn
&=\frac{1}{1-\bar{\gamma}\bar{\mu}}
=\frac{1}{\mu+\gamma\bar\mu}
=\frac{\pb}{\gamma\bar{\mu}}.
\end{align}

\subsection{Proof of Lemma~\ref{lemma_EZ}}\label{appendix_EZ}
The mean waiting time $\mathbb{E}[Z_i]$ is derived as
\allowdisplaybreaks
\begin{align}
\mathbb{E}[Z_i]&=\displaystyle\sum_{j=1}^{\infty}j\Pr(Z_i=j) \overset{(a)}{=}\displaystyle\sum_{j=1}^{\beta-1}(\beta-j)\Pr(Y_{i-1}=j)\nn
&\overset{}{=}\displaystyle\beta\sum_{j=1}^{\beta-1}\bar\mu^{j-1}\mu-\sum_{j=1}^{\beta-1}j\bar\mu^{j-1}\mu\nn
&=\displaystyle\frac{\beta\mu}{\bar\mu}\sum_{j=1}^{\beta-1}\bar\mu^{j}-\frac{\mu}{\bar\mu}\sum_{j=1}^{\beta-1}j\bar\mu^{j}\overset{(b)}{=}\displaystyle\frac{1}{\mu}\left( \beta\mu + \bar\mu^\beta -1 \right),
\end{align}
where $(a)$ follows from ${Z_i=f(Y_{i-1})=(\beta-Y_{i-1})^+}$ and $(b)$ follows from the identities \cite[Page~484]{yates1999probability}
\begin{align}
\label{series_alpha_j_low}
\sum_{j=1}^{m}\alpha^j&=\dfrac{\alpha(1-\alpha^m)}{1-\alpha},\\
\label{series_j_alpha_j_low}
\sum_{j=1}^{m}j\alpha^j&=\dfrac{\alpha(1-\alpha^m(1+m(1-\alpha)))}{(1-\alpha)^2}.
\end{align}

\subsection{Proof of Lemma~\ref{lemma_EZY}}\label{appendix_EZY}
Since \eqref{threshold-f}
implies $\Pr(Z_i=k \mid Y_{i-1}=j)=1_{\{k=(\beta-j)^+\}}$,
\begin{align}
\E{Z_iY_{i-1}}
&=\sum_{j=1}^{\infty}\sum_{k=1}^{\infty}jk\Pr(Y_{i-1}=j)\Pr(Z_i=k \mid Y_{i-1}=j)\nn
&=\sum_{j=1}^{\beta-1}j(\beta-j)\Pr(Y_{i-1}=j).
\end{align}
This implies
\begin{align}
\E{Z_iY_{i-1}}
&=\beta\sum_{j=1}^{\beta-1}j\bar\mu^{j-1}\mu - \sum_{j=1}^{\beta-1}j^2\bar\mu^{j-1}\mu\nn
&=\frac{\beta\mu}{\bar\mu}\sum_{j=1}^{\beta-1}j\bar\mu^{j} - \frac{\mu}{\bar\mu}\sum_{j=1}^{\beta-1}j^2\bar\mu^{j}\nn
&\overset{(a)}{=}\displaystyle\frac{1}{\mu^2}\left( \mu(1 + \beta)+ \bar\mu^\beta (2-\mu+ \beta\mu)- 2 \right),
\end{align}
where  $(a)$ follows from \eqref{series_j_alpha_j_low} and the identity 
\begin{align}
\sum_{j=1}^{m}j^2\alpha^j&=\dfrac{\alpha^2+\alpha+\alpha^{m+2}(2m^2+2m-1)}{(1-\alpha)^3}\nn
&\qquad -\frac{\alpha^{m+1}[(m+1)^2+m^2\alpha^{2}]}{(1-\alpha)^3}.
\label{series_j2_alpha_j_low}
\end{align}
Note that \eqref{series_j2_alpha_j_low} can be verified by taking derivatives with respect to $\alpha$ of both sides of \eqref{series_j_alpha_j_low}. 
 
\subsection{Proof of Lemma~\ref{lemma_EZ2}}\label{appendix_EZ2}
It follows from \eqref{threshold-f} that 
\allowdisplaybreaks
\begin{align}
&\mathbb{E}[Z_i^2]=\displaystyle\sum_{j=1}^{\infty}j^2\Pr(Z_i=j)\overset{(a)}{=}\displaystyle\sum_{j=1}^{\beta-1}(\beta-j)^2\Pr(Y_{i-1}=j)\nn
&\nonumber\overset{}{=}\displaystyle\beta^2\sum_{j=1}^{\beta-1}\bar\mu^{j-1}\mu-2\beta\sum_{j=1}^{\beta-1}j\bar\mu^{j-1}\mu  \sum_{j=1}^{\beta-1}j^2\bar\mu^{j-1}\mu\\
&\overset{(a)}{=}\displaystyle-\frac{1}{\mu^2}\left(\mu(1 + 2\beta) + \bar\mu^\beta(2-\mu) - \beta^2\mu^2 - 2\right).
\end{align} 
Note that $(a)$ follows from 
 \eqref{series_alpha_j_low}, \eqref{series_j_alpha_j_low}, and \eqref{series_j2_alpha_j_low}. 
  \subsection{Proof of Lemma~\ref{lemma_beta_bound}}\label{appendixBBeta}
Let $\beta^*$ denote the optimal $\beta$ for the $\Wone$ policy. To find an upper bound on 
$\beta^*$, we observe that the $\Wone$ policy with $\beta^*$ results in an average AoI not greater than that of the $\ZWone$ policy. Thus
$\beta^*$ satisfies the following inequality
\begin{align}\label{r_L_01}
  \Delta^{\ZWone}- \Delta^{\Wone}(\beta^*)\ge0,
\end{align}
where $\Delta^{\Wone}(\beta^*)$ represents the average AoI under the $\Wone$ policy with ${\beta=\beta^*}$. By substituting the expressions for $\Delta^{\ZWone}$ in \eqref{ZW1exp} and $\Delta^{\Wone}(\beta^*)$ in \eqref{eq_Wait1_beta} into \eqref{r_L_01}, some algebra shows that $\beta^*$ must satisfy the constraint
\begin{align}\label{r_L_02}
q(\beta^*)+\bar{\mu}^{\beta^*}[2\beta^*(\mu+\gamma)+2]\le 0,
\end{align}
where $q(\beta)$ denotes the quadratic function
\begin{align}\label{qbeta1}
    q(\beta)&\equiv  \beta^2(\mu^2+\gamma\mu)
    +\beta(\mu^2+\gamma\mu-2\gamma)-2.
\end{align}
We observe that $q(0)=-2$ and that $q(\beta)>0$ for sufficiently large $\beta$. 
Since the term $\bar{\mu}^{\beta^*}[2\beta^*(\mu+\gamma)+2]$ on the left side of \eqref{r_L_02} is positive, it follows from \eqref{r_L_02}  that $\beta^*$ must satisfy $\beta^*\le \beta^{\max}$ where $\beta^{\max}$ is the maximum value of $\beta$ satisfying $q(\beta)\le 0$.
With $a=\mu^2+\gamma\mu$, the quadratic formula implies
\begin{IEEEeqnarray}{rCl}\label{r_L_04}
\beta^{\max}&=&\dfrac{2\gamma-a+
\sqrt{(a-2\gamma)^2+8a}}{2a}.\IEEEeqnarraynumspace
\end{IEEEeqnarray}
The claim follows with substitution of $a$ and  the observation that $\beta^*$ is an integer.

\bibliographystyle{IEEEtran}
\bibliography{Bibliography}

\end{document}